\pdfoutput=1
\RequirePackage{ifpdf}
\ifpdf 
\documentclass[pdftex]{sigma}
\else
\documentclass{sigma}
\fi

\usepackage{mathtools}
\usepackage{pgf,tikz}
\usetikzlibrary{arrows,angles,quotes,decorations.markings}
\tikzstyle directed=[postaction={decorate,decoration={markings,
mark=at position .65 with {\arrow{latex}}}}]

\newcommand{\dS}{\frac{{\rm d}}{{\rm d}S}}

\newcommand{\ubr}{[U]}
\newcommand{\dSSS}{\frac{{\rm d}^{3}}{{\rm d}S^{3}}}
\newcommand{\dSS}{\frac{{\rm d}^{2}}{{\rm d}S^{2}}}
\newcommand{\acomw}{[W,\cdot]_{+} }
\newcommand{\comw}{[W,\cdot] }
\newcommand{\dinv}{\frac{{\rm d}}{{\rm d}S}^{-1}}
\DeclareMathOperator{\Tr}{Tr}
\DeclareMathOperator{\Res}{Res}
\DeclareMathOperator{\diag}{diag}

\numberwithin{equation}{section}

\newtheorem{Theorem}{Theorem}[section]
\newtheorem{Corollary}[Theorem]{Corollary}
\newtheorem{Lemma}[Theorem]{Lemma}
\newtheorem{Proposition}[Theorem]{Proposition}
 { \theoremstyle{definition}
\newtheorem{Definition}[Theorem]{Definition}
\newtheorem{Example}[Theorem]{Example}
\newtheorem{Remark}[Theorem]{Remark} }
\newtheorem{pb}[Theorem]{Problem}

\begin{document}

\newcommand{\arXivNumber}{2007.05707}

\renewcommand{\PaperNumber}{002}

\FirstPageHeading

\ShortArticleName{A Fully Noncommutative Painlev\'e II Hierarchy: Lax Pair and Solutions}

\ArticleName{A Fully Noncommutative Painlev\'e II Hierarchy: Lax\\ Pair and Solutions Related to Fredholm Determinants}

\Author{Sofia TARRICONE~$^{\dag\ddag}$}

\AuthorNameForHeading{S.~Tarricone}

\Address{$^\dag$~LAREMA, UMR 6093, UNIV Angers, CNRS, SFR Math-Stic, France}
\EmailD{\href{mailto:tarricone@math.univ-angers.fr}{tarricone@math.univ-angers.fr}}

\Address{$^\ddag$~Department of Mathematics and Statistics, Concordia University,\\
\hphantom{$^\ddag$}~1455 de Maisonneuve W., Montr\'eal, Qu\'ebec, Canada, H3G 1M8}

\ArticleDates{Received July 25, 2020, in final form December 31, 2020; Published online January 05, 2021}

\Abstract{We consider Fredholm determinants of matrix Hankel operators associated to matrix versions of the $n$-th Airy functions. Using the theory of integrable operators, we relate them to a fully noncommutative Painlev\'e II hierarchy, defined through a matrix-valued version of the Lenard operators. In particular, the Riemann--Hilbert techniques used to study these integrable operators allows to find a Lax pair for each member of the hierarchy. Finally, the coefficients of the Lax matrices are explicitly written in terms of the matrix-valued Lenard operators and some solutions of the hierarchy are written in terms of Fredholm determinants of the square of the matrix Airy Hankel operators.}


\Classification{34M56; 35Q15; 47B35; 33C10}


\section{Introduction}
The aim of this work is to relate a family of solutions of a noncommutative version of the Painlev\'e II hierarchy to Fredholm determinants of a matrix version of the $ n $-th Airy Hankel operators. The scalar versions of these operators have been recently studied in \cite{le2018multicritical}, in relation with determinantal point processes (we will discuss this relation later on).

 In order to construct our matrix analogue, we first define a matrix-valued version of the $ n $-th Airy function, in the following way
\begin{equation}\label{eq:matrixairy}
\mathbf{Ai}_{2n+1}(x,\vec{s})= \big( c_{j,k} \mathrm{Ai}_{2n+1}(x+s_{j}+s_{k})\big) _{j,k=1}^{r}, \qquad c_{j,k} \in \mathbb{C}, \qquad x \in \mathbb{R} ,
\end{equation}
where $ \mathrm{Ai}_{2n+1}(x+s_{j}+s_{k}) $ is a shift of the $ n $-th scalar Airy function, for some real parameters~$ s_{l}$, $l=1,\dots, r $. We recall that the $ n$-th scalar Airy function, $ \mathrm{Ai}_{2n+1}(x) $, is defined as a particular solution of the differential equation
\begin{equation}\label{eq:genairyeq}
\frac{{\rm d}^{2n}\phi(x)}{{\rm d} x^{2n}}=(-1)^{n+1}x\phi(x)
\end{equation}
for each $n\geq 1$. We refer to~\cite{miura1985particular} for details about the solutions of the generalized Airy equations~\eqref{eq:genairyeq}. In this paper we will consider these functions $ \mathrm{Ai}_{2n+1}(x) $ as contour integrals
\begin{equation*}
\mathrm{Ai}_{2n+1}(x) := \int_{\gamma^{n}_{+}}\frac{1}{2\pi }\exp\left( \frac{{\rm i} \mu^{2n+1}}{2n+1}+{\rm i}x\mu\right){\rm d}\mu ,\qquad x \in \mathbb{R},
\end{equation*} for $\gamma^{n}_{+}$ an appropriate curve, which we will specify later on.

With the matrix-valued Airy functions $ \mathbf{Ai}_{2n+1}(x,\vec{s}) $ defined in~\eqref{eq:matrixairy}, the matrix Airy Hankel operators $\mathcal{A}{\rm i}_{2n+1}$ are defined in the standard way
\begin{equation}\label{eq:conop}
\left( \mathcal{A}{\rm i}_{2n+1}\mathbf{f}\right)(x):=\int_{\mathbb{R}_{+}} \mathbf{Ai}_{2n+1}(x+y,\vec{s})\mathbf{f}(y)\,{\rm d}y ,
\end{equation}
for any $ \mathbf{f}=(f_{1},\dots,f_{r})^{\rm T} \in L^{2}\big(\mathbb{R}_{+},\mathbb{C}^{r} \big) $.
It is actually on the square of this sequence of operators that we focused our study, and in particular on the Fredholm determinants defined as
\begin{equation}\label{eq:freddet}
F^{(n)} (s_{1},\dots, s_{r} ) := \det\big( {\rm Id}_{\mathbb{R}_{+}}- \mathcal{A}{\rm i}_{2n+1}^{2} \big) ,
\end{equation}
that are well defined since the operators $ \mathcal{A}{\rm i}_{2n+1}$ are trace-class on $L^{2} \big( \mathbb{R}_+,\mathbb{C}^r \big) $ (as follows from Proposition~2.1 in~\cite{bertola2012fredholm}).

 The core of this work is devoted to establish a relation between the Fredholm determinants~\eqref{eq:freddet} and some solution of a noncommutative Painlev\'e II hierarchy.
In particular, the results obtained in~\cite{bertola2012fredholm}, where the authors extend the theory of integrable operators of Its--Izergin--Korepin--Slavnov~\cite{its1990differential}, can be directly applied to the matrix operators $ \mathcal{A}{\rm i}_{2n+1} $ defined in~\eqref{eq:conop}. As byproduct, an equality between the Fredholm determinants $ F^{(n)} (s_{1},\dots, s_{r} ) $ and those of certain integrable operators can be established.
The study of these integrable operators involves particular Riemann--Hilbert problems (defined in Riemann--Hilbert Problem~\ref{pb:rhp}), and these are indeed the main tool used in the following. In particular: starting from them we construct a~solution for the isomonodromic Lax pair of the noncommutative Painlev\'e~II hierarchy, that we are going to define as follows.

To start with, we first define a sequence of matrix-valued differential polynomials, by using a noncommutative version of the well known Lenard operators $ \mathcal{L}_{n} $, used to define the scalar Painlev\'e II hierarchy. In the following, $ U$, $W $
are functions depending on all the parameters $ s_{l}$, $l=1,\dots,r $ with values in $ \mathrm{Mat}(r\times r,\mathbb{R}) $. The symbols $ [\, ,\, ]$ and $ [\, ,\, ]_{+} $ indicate respectively the standard commutator and anti-commutator between two matrices, since differential polynomials in~$U$ are noncommutative quantities.

 Then each differential polynomial $ \mathcal{L}_{n} \ubr $ is defined by the following recursive relation
\begin{gather}
\mathcal{L}_{0}\ubr= \frac{1}{2}I_{r}, \nonumber \\
\dS \mathcal{L}_{n}\ubr=\left( \dSSS + [ U,\cdot ] _+ \dS +\dS [ U,\cdot ] _+ + [ U,\cdot ] \dinv [ U,\cdot ] \right) \mathcal{L}_{n-1}\ubr , \qquad n \geq 1,\label{eq:lenardrec}
\end{gather}
where the differential operator $\dS$ is defined as
\begin{equation}\label{eq:ds}
\dS:=\sum_{ k=1}^{r}\frac{\partial}{\partial s_{k}},
\end{equation}
and $\dinv$ in intended as the corresponding formal antiderivative.
The recursive relation for the noncommutative version of the Lenard operators~$\mathcal{L}_{n}$, $n\geq 1 $, is related to the recursion operator for the noncommutative KdV equation, introduced in~\cite{olver1998integrable}. There the authors already conjectured about the locality of these operators computed in~$U$, but the formal proof of that was done some years later in~\cite{olver2000classification} (Theorem~6.2 in this last paper).

Finally we define our noncommutative Painlev\'e II hierarchy as follows
\begin{equation}\label{eq:Piiintr}
\mathrm{PII^{(n)}_{\rm NC}}\colon \ \left( \dS + \acomw \right) \mathcal{L}_{n}\ubr =(-1)^{n+1}4^{n} [S, W ]_{+} ,
\end{equation}
where $U \coloneqq\dS W-W^{2}$ is the Miura transform of $W$ and the variable $ S $ is the diagonal matrix $ S:= \diag(s_{1}, \dots, s_{r})$ so that the anti-commutator in the right hand side is needed (also note that $ \dS S=I_{r} $). For this reason we refer to our hierarchy as a fully noncommutative one, since in its definition~\eqref{eq:Piiintr} also the independent variable~$S$ is noncommutative. A matrix Painlev\'e~II hierarchy, constructed by using a noncommutative version of Lenard operators as in~\eqref{eq:lenardrec}, was recently studied in~\cite{gordoa2016matrix} but in this paper the independent variable is a scalar.

In this work, first of all, we found out that the hierarchy~\eqref{eq:Piiintr} admits an isomonodromic Lax pair with Lax matrices that are block-matrices of dimension~$ 2r $. Furthermore, they are explicitly written in terms of the matrix-valued Lenard operators defined in \eqref{eq:lenardrec}. The result proved in Section~\ref{sec:IV} is summarized in the following proposition.
\begin{Proposition}\label{Proposition:laxpairintro}
For each fixed $ n $ there exist two polynomial matrices in $ \lambda $, namely $ L^{(n)}$, $M^{(n)} $, respectively of degree $ 1 $ and $ 2n $, such that the following system
\begin{gather}
\dS\Psi ^{(n)}(\lambda, \vec{s})= L^{(n)}(\lambda, \vec{s})\Psi ^{(n)} ( \lambda, \vec{s} ), \nonumber\\
\frac{\partial}{\partial {\lambda}}\Psi ^{(n)}(\lambda, \vec{s})=M^{(n)}(\lambda, \vec{s})\Psi ^{(n)} ( \lambda, \vec{s} ) \label{eq:laxpairRHintro}
\end{gather}
is an isomonodromic Lax pair for the $ n$-th equation of the matrix Painlev\'e II hierarchy \eqref{eq:Piiintr}.

In particular the matrices $ L^{(n)}$, $M^{(n)} $ have the following forms
\begin{equation*}
L^{(n)}(\lambda, \vec{s})= \begin{pmatrix}
{\rm i}\lambda I_{r}&W(\vec{s})\\
W(\vec{s})& - {\rm i}\lambda I_{r}
\end{pmatrix},
\end{equation*}
and
\begin{equation*}
M^{(n)}(\lambda, \vec{s})=\begin{pmatrix}
A(\lambda,\vec{s})+{\rm i}S&{\rm i}G(\lambda,\vec{s})\\
-{\rm i}G(\lambda,\vec{s})&- A(\lambda,\vec{s})-{\rm i}S
\end{pmatrix}+\begin{pmatrix}
E(\lambda,\vec{s})&F(\lambda,\vec{s})\\
F(\lambda,\vec{s})&E(\lambda,\vec{s})
\end{pmatrix},
\end{equation*}
where
\begin{gather*}
A(\lambda,\vec{s})=\sum_{ k=0}^{n}\frac{{\rm i}}{2} \lambda ^{2n-2k}A_{2n-2k} (\vec{s} ), \qquad \text{with} \quad A_{2n}=I_{r},\\
G(\lambda,\vec{s})=\sum_{k=1}^{n}\frac{{\rm i}}{2}\lambda ^{2n-2k}G_{2n-2k} ( \vec{s} ), \\
E(\lambda,\vec{s})=\sum_{k=1}^{n} \frac{{\rm i}}{2}\lambda ^{2n-2k+1}E_{2n-2k+1} ( \vec{s} ) , \\
F(\lambda,\vec{s})=\sum_{k=1}^{n}\frac{{\rm i}}{2}\lambda ^{2n-2k+1}F_{2n-2k+1} ( \vec{s} ) .
\end{gather*}
All the coefficients $ A_{2n-2k}$, $G_{2n-2k}$, $E_{2n-2k+1}$, $F_{2n-2k+1} $ are expressed in terms of the Lenard operators through the formulae~\eqref{eq:coeffM}.
\end{Proposition}

This result can be thought as the noncommutative analogue of the well known isomonodromic Lax pair for the scalar Painlev\'e~II hierarchy studied in \cite{clarkson2006lax}, and resulting from a self-similarity reduction of the Lax pair for the modified KdV hierarchy.

Finally, we construct a solution $ \Psi^{(n)} $ for the Lax pair~\eqref{eq:laxpairRHintro}, by using the solution of the Riemann--Hilbert Problem~\ref{pb:rhp} involved in the study of the integrable operators associated to the matrix operators squared $ \mathcal{A}{\rm i}_{2n+1}^{2} $.

As byproduct, we obtain the following relation between some solutions of the hierarchy~\eqref{eq:Piiintr} and the Fredholm determinants \eqref{eq:freddet}. This is indeed the final result of this work and it is proved at the end of Section \ref{sec:IV}.
\begin{Corollary}\label{Corollary:sols}
There exists a solution $ W $ of the $ n$-th member of the matrix PII hierarchy~\eqref{eq:Piiintr}, that is connected to the Fredholm determinant of the $ n$-th Airy matrix Hankel operator through the following formula
\begin{equation*}
-\Tr\big( W^{2} (\vec{s})\big) = \dSS \ln \big( F^{(n)} (s_{1},\dots, s_{r} ) \big) .
\end{equation*}
Defining $\mathrm{s} \coloneqq \frac{1}{r}\sum_{j=1}^{r}s_j$, and $\delta_j\coloneqq s_j - \mathrm{s}$, this solution $ W $ in the regime $\mathrm{s}\rightarrow + \infty$ with $ | \delta_j | \leq m $ for every~$ j $, has asymptotic behavior $( W) _{k,l=1}^r \sim -2 ( c_{kl}\mathrm{Ai}_{2n+1} (s_k + s_l) )_{k,l=1}^r $.
\end{Corollary}

We remark that
in \cite{bertola2012fredholm} the above result was actually proved for the first equation of the hierarchy, i.e., for the case $n=1$. In this paper we extend the result to the all hierarchy \eqref{eq:Piiintr}, i.e., to every $n \in \mathbb{N}$.
This result is a generalization of well known results in the scalar case. Indeed, in this case the $ n $-th Airy kernels, defined as
\begin{equation}\label{eq:airykern}
K_{\mathrm{Ai}_{2n+1}}(x,y):= \int_{\mathbb{R}_{+}}\mathrm{Ai}_{2n+1}(x+z)\mathrm{Ai}_{2n+1}(z+y){\rm d}z,
\end{equation}
are known to be related to Painlev\'e trascendents. For $ n=1 $, it was shown in \cite{tracy1994level} that the Airy kernel \eqref{eq:airykern} is related to the Hastings--McLeod solution of the Painlev\'e II equation (discussed for example in \cite{clarkson1988connection,deift1995asymptotics,hastings1980boundary}). This study turns out to be very interesting since it creates a connection between integrable systems and determinantal point processes. Indeed, the Airy kernel defines the so called Airy point process, that arises in many areas of mathematics, such as statistical mechanics and random matrix theory (here some examples of literature \cite{johansson2000shape, tracy1994level,tracy2002distribution}).

For Airy kernels \eqref{eq:airykern} with $ n> 1 $, a generalization of this kind of results has been recently studied in \cite{le2018multicritical}, and in \cite{cafasso2019fredholm}.

In this work, we see that the matrix Airy Hankel operators squared $\mathcal{A}{\rm i}_{2n+1} ^{2}$ can actually be interpreted as kernels for determinantal point processes on the space of configuration $ \lbrace 1,\dots,r \rbrace \times \mathbb{R} $ (under certain assumptions on the matrix $ C=(c_{j,k})_{j,k=1}^{r} $), and it would be interesting to study whether they describe phenomena in random matrix theory or statistical mechanics.

Here is a more precise list of what it is done in this work.
\begin{itemize}\itemsep=0pt
\item In Section \ref{sec:II} the general theory developed in \cite{bertola2012fredholm} is applied to the operators $\mathcal{A}{\rm i}_{2n+1}^{2} $, in order to associate the Fredholm determinants \eqref{eq:freddet} to the ones of certain integrable operators. The most important consequence of this study is indeed Theorem \ref{Proposition:ascoeffredolm}, that establishes a~relation between Fredholm determinant~\eqref{eq:freddet} and the solution of Riemann--Hilbert Problem~\ref{pb:rhp}. Furthermore, in this section it is provided in which hypothesis the solution exists (Theorem~\eqref{Theorem:existence solution rhp}), and so the relation for the Fredholm determinants found in Theorem~\ref{Proposition:ascoeffredolm} holds.
\item In Section \ref{sec:III} the fully noncommutative Painlev\'e~II hierarchy is introduced and the first equations are explicitly written.
\item In the first part of Section~\ref{sec:IV}, the proof of Proposition~\ref{Proposition:laxpairintro} is given and the construction of the solution $ \Psi^{(n)} $ of the isomonodromic Lax pair~\eqref{eq:laxpairRHintro} for the hierarchy~\eqref{eq:Piiintr} is implemented. Finally in the end of Section~\ref{sec:IV}, Corollary~\ref{Corollary:sols} is proved, by using Theorem~\ref{Proposition:ascoeffredolm} and the properties of the solution~$ \Psi^{(n)} $ of the isomonodromic Lax pair~\eqref{eq:laxpairRHintro}.
\end{itemize}

\section[Riemann Hilbert problems associated to the matrix Airy operators]{Riemann Hilbert problems associated\\ to the matrix Airy operators}\label{sec:II}

In this section we are going to study the Fredholm determinants of some matrix-valued Airy Hankel operator. By using the theory developed in \cite{bertola2012fredholm} we can associate to this sequence of matrix Airy operators a sequence of integrable operators with certain kernels, such that their Fredholm determinants are equal.
Properties of this kind of integral kernels are studied through Riemann--Hilbert problems. As byproduct, this procedure allows to find the fundamental relation between Fredholm determinants of the Airy matrix Hankel operators and the first asymptotic coefficient of the solutions of these Riemann--Hilbert problems, as proved in Theorem \ref{Proposition:ascoeffredolm}.

To start with, we recall some basic fact about the scalar generalized Airy functions $ \mathrm{Ai}_{2n+1} $. For each $ n \in \mathbb{N} $, we consider these functions $ \mathrm{Ai}_{2n+1} $ as the contour integrals
\begin{equation}\label{eq:int representation airy}
\mathrm{Ai}_{2n+1}(x) := \int_{\gamma^{n}_{+}}\frac{1}{2\pi }\exp\left( \frac{{\rm i} \mu^{2n+1}}{2n+1}+{\rm i}x\mu\right){\rm d}\mu ,\qquad x \in \mathbb{R},
\end{equation}
where $ \gamma^{n}_{\pm}$ are curves in the upper (lower) complex plane with asymptotic rays at $\pm \infty $ that are $ \phi_{\pm}^{n}:= \frac{\pi}{2} \pm \frac{\pi n}{2n + 1} $, and such that $ \gamma^{n}_{-}=- \gamma^{n}_{+} $. An example of these curves for $ n=1 $ is given in Fig.~\ref{fig:curves3}.
\begin{Definition}
The $ n $-th matrix-valued Airy function is defined as
\begin{equation*}
\mathbf{Ai}_{2n+1}(x,\vec{s}):= \big( c_{j,k} \mathrm{Ai}_{2n+1}(x+s_{j}+s_{k})\big) _{j,k=1}^{r}, \qquad x \in \mathbb{R}.
\end{equation*}
Here $ C= ( c_{j,k}) _{j,k=1}^{r} \in \mathrm{Mat}( {r}\times r,\mathbb{C}) $ and the parameters $s_l \in \mathbb{R}$, $l=1,\dots,r$.
\end{Definition}

With these functions we construct the matrix-valued operators we are going to study in the following.
\begin{Definition}
We consider $\lbrace \mathcal{A}{\rm i}_{2n+1}\rbrace_{n \in \mathbb{N}} $ the sequence of matrix Hankel operators acting on any $ \mathbf{f}=(f_{1},\dots,f_{r})^{\rm T} \in L^{2} \big(\mathbb{R}_{+},\mathbb{C}^{r} \big) $ s.t.
\begin{equation}\label{eq:convop}
 ( \mathcal{A}{\rm i}_{2n+1}\mathbf{f} )(x):=\int_{\mathbb{R}_{+}} \mathbf{Ai}_{2n+1}(x+y,\vec{s})\mathbf{f}(y)\,{\rm d}y .
\end{equation}
\end{Definition}

Component wise the $ n$-th Hankel operator $\mathcal{A}{\rm i}_{2n+1}, $ looks like
\begin{equation}\label{eq:comp conv opera}
\left( \mathcal{A}{\rm i}_{2n+1}\mathbf{f}\right)_{j}(x)=\sum_{ k=1}^{r}c_{j,k} \int_{\mathbb{R}_{+}} \mathrm{Ai}_{2n+1}(x+y+s_{j}+s_{k})f_{k}(y)\,{\rm d}y, \qquad j=1,\dots,r.
\end{equation}

\begin{Remark}One can equivalently define the matrix-valued generalized Airy functions as contours integrals, in the following way.
For each $ n\in \mathbb{N} $
\begin{itemize}\itemsep=0pt
\item we take $ s_{1}, \dots, s_{r} $ real parameters and $ S:=\diag(s_{1}, \dots, s_{r}) $ and we define the matrix-valued complex function
\begin{equation}\label{eq:theta}
\theta_{2n+1}(\mu,\vec{s}):=\frac{{\rm i} \mu^{2n+1}}{2(2n+1)}I_{r}+ {\rm i}\mu S,
\end{equation}
where $ I_{r} $ is the identity matrix of dimension $ r $.
\item Then, we take the matrix $ C = ( c_{j,k} ) _{j,k=1}^{r} \in \mathrm{Mat} ( {r}\times r,\mathbb{C} ) $ we define the matrix-valued function
\begin{equation}\label{eq:funcrjump}
r^{(n)} (\lambda,\mu,\vec{s}):= \frac{1}{2\pi {\rm i} }\exp (\theta_{2n+1}(\lambda,\vec{s})) C \exp(\theta_{2n+1}(\mu,\vec{s})).
\end{equation}
\item Finally, we can define the generalized matrix Airy function as
\begin{equation*}
\mathbf{Ai}_{2n+1}(x,\vec{s})=\int_{\gamma^{n}_{+}}{\rm i} r^{(n)} (\mu,\mu,\vec{s}) \exp({\rm i}x\mu)\,{\rm d}\mu,
\end{equation*}
where the integral is computed entry by entry.
\end{itemize}
\end{Remark}

\begin{figure}[t]\centering
\begin{tikzpicture}[scale=5,
axis/.style={very thick, ->, >=stealth'},
important line/.style={thick},
dashed line/.style={dashed, thin},
pile/.style={thick, ->, >=stealth', shorten <=2pt, shorten
>=2pt},
every node/.style={color=black}]

\draw[axis] (-1,0) -- (1,0);
\draw[axis] (0,-0.5) -- (0,0.5);
\draw[>=latex,directed,blue] (-0.866, 0.52)..controls (0,0.000000001).. (0.866,0.52) node[right,blue]{$\gamma^{3}_{+}$};
\draw[>=latex] (0, 0) to (0.866, 0.5);
\draw[>=latex] (0, 0) to (-0.866, -0.5);
\draw[>=latex] (0, 0) to (-0.866, 0.5);
\draw[>=latex] (0, 0) to (0.866, -0.5);
\draw[>=latex,directed,blue] (0.866, -0.52)..controls (0,-0.000000001).. (-0.866,-0.52) node[left,blue]{$\gamma^{3}_{-}$};
\coordinate (a) at (0,0);
\coordinate (b) at (1.732,1);
\coordinate (c) at (1.732,0);
\pic[draw, ->, "$\phi_{+}^3$", angle eccentricity=2] {angle = c--a--b};
\end{tikzpicture}
\caption{These are the contours $ \gamma_{\pm}^{3} $ for the integral representation \eqref{eq:int representation airy} of the Airy function $\mathrm{Ai}_{3}$ (case $n=1$). Their asymptotics at $\pm \infty $ are $ \phi_{\pm}^{3}:= \frac{\pi}{6},\frac{5\pi}{6}$.}
\label{fig:curves3}
\end{figure}
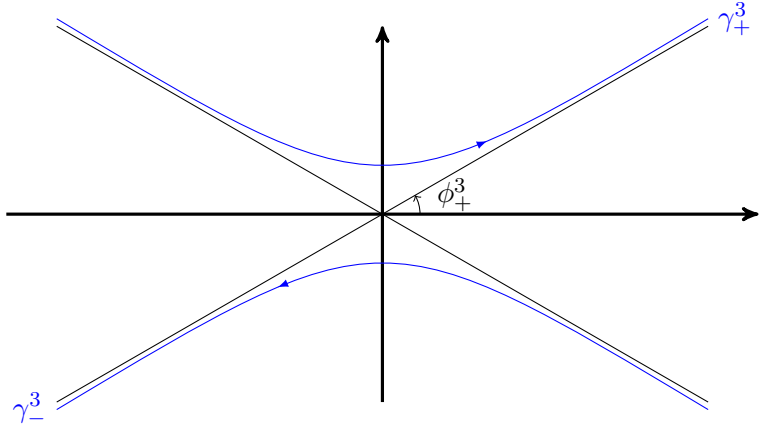

We are actually interested in the square of the matrix Airy operators defined above in~\eqref{eq:convop}. Indeed, the Fredholm determinants of these squared operators will be related to the noncommutative Painlev\'e~II hierarchy~\eqref{eq:Piiintr}.

We are now going to define a sequence of Riemann--Hilbert problems related to the matrix Airy operators. These are indeed the building blocks necessary to find the relation between Fredholm determinants of the matrix Airy operators and our noncommutative Painlev\'e~II hierarchy.

\begin{Remark}From now on, in order to simplify the notation, the dependence on $\vec{s}$ in the quantities \eqref{eq:theta}, \eqref{eq:funcrjump} will be omitted and we will use the abbreviation $r^{(n)} (\lambda,\lambda,\vec{s})=r^{(n)} (\lambda)$.
\end{Remark}

\begin{pb}\label{pb:rhp}
Find a $ (\lambda-) $analytic matrix-valued function \[\Xi^{(n)}(\lambda) \colon \ \mathbb{C}\setminus \big( \gamma^{n}_{+} \cup \gamma^{n}_{-}\big)\rightarrow {\rm GL}(2r, \mathbb{C}) , \]
admitting continuous extension to the contour $ \gamma^{n}_{+} \cup \gamma^{n}_{-}$ from either side and such that it satisfies the following two conditions:
\begin{itemize}\itemsep=0pt
\item the jump condition for each $ \lambda \in \gamma^{n}_{+} \cup \gamma^{n}_{-} $
\begin{equation}\label{eq:jump}
\Xi^{(n)}_{+}(\lambda)=\Xi_{-} ^{(n)}(\lambda)
\underbrace{\begin{pmatrix}
I_{r} & -2\pi {\rm i} r^{(n)} (\lambda)\chi_{\gamma^{n}_{+}}(\lambda)\\
-2\pi {\rm i} r^{(n)} (-\lambda)\chi_{\gamma^{n}_{-}}(\lambda)& I_{r}
\end{pmatrix}}_{:=J^{(n)}(\lambda,\vec{s})}
,
\end{equation}
where we denote by $ \Xi^{(n)}_{\pm} $ the boundary values of $ \Xi^{(n)} $ for $ \lambda \in \gamma_{+}^{(n)} \cup \gamma_{-}^{(n)} $, approaching the boundary from the left $ (+) $ and the right $ (-) $ nontangentially.
\item the asymptotic condition for $|\lambda| \rightarrow \infty $
\begin{equation}\label{eq:asymp}
\Xi^{(n)}(\lambda) \sim I_{2r } + \sum_{j\geq 1 } \frac{\Xi^{(n)}_{j}}{\lambda^{j}}.
\end{equation}
\end{itemize}
\end{pb}

\begin{Remark}In the following we are going to use the Pauli's tensorized matrices, that have the same property as the ones in the usual Clifford algebra. In particular we denote the tensorized matrices by
\begin{equation*}
\hat{\sigma}_{1}=\sigma_{1} \otimes I_{2r},\qquad
\hat{\sigma}_{2}=\sigma_{2} \otimes I_{2r}, \qquad \hat{\sigma}_{3}=\sigma_{3} \otimes I_{2r},
\end{equation*}
where
\begin{equation*}
\sigma_{1}=\begin{pmatrix}
0&1\\1&0
\end{pmatrix}, \qquad \sigma_{2}=\begin{pmatrix}
0&{\rm i}\\
-{\rm i}&0
\end{pmatrix}, \qquad \sigma_{3}=\begin{pmatrix}
1&0\\
0&-1
\end{pmatrix}.
\end{equation*}
Then the standard relations hold also in this case:
\begin{equation*}
[ \hat{\sigma}_{1},\hat{\sigma}_{2}] =-2{\rm i}\hat{\sigma}_{3},
\qquad
[ \hat{\sigma}_{1},\hat{\sigma}_{3}] =2{\rm i}\hat{\sigma}_{2},
\qquad
[ \hat{\sigma}_{2},\hat{\sigma}_{3}] =-2{\rm i}\hat{\sigma}_{1},
\qquad \hat{\sigma}_{i}^{2}=I_{2r}, \quad \forall\, i.
\end{equation*}
\end{Remark}

The following symmetry property will be useful in the next computations.
\begin{Corollary}
The asymptotic coefficients appearing in equation \eqref{eq:asymp} have the following form
\begin{gather}
\Xi^{(n)}_{2j}=\alpha_{2j}^{(n)} \otimes I_{2}+\beta_{2j}^{(n)}\otimes\sigma_{1},\nonumber\\
\Xi^{(n)}_{2j-1}=\alpha_{2j-1}^{(n)} \otimes \sigma_{3}+\beta_{2j-1}^{(n)}\otimes\sigma_{2},
\qquad j\geq 1.\label{eq:symcoef}
\end{gather}
Here $\alpha_{l}^{(n)}$, $\beta_{l}^{(n)}$ for every $l\geq 1$ correspond to the $r\times r$ matrices in the entries $(1,1)$ and $(1,2)$ of the block matrix~$\Xi_l^{(n)}$.

An analogue statement is true for the asymptotic coefficients of the inverse of the solution of the Riemann--Hilbert Problem~{\rm \ref{pb:rhp}}, namely $ \Theta^{(n)}:= \big( \Xi^{(n)}\big)^{-1} $.
\end{Corollary}
\begin{proof}We first prove the symmetry condition for the asymptotic coefficients of $ \Xi^{(n)}. $
We start observing that the jump matrix $ J^{(n)} $ for $\lambda \in \gamma^{n}_{+} \cup \gamma^{n}_{-}$ has the following symmetry
\begin{equation*}
\hat{\sigma}_{1}J^{(n)}(\lambda,\vec{s})\hat{\sigma}_{1}=J^{(n)}(-\lambda,\vec{s}),
\end{equation*}
just using the definition of $ \gamma^{n}_{-}=-\gamma^{n}_{+} $.
This directly implies that also the solution of the Riemann--Hilbert Problem~\ref{pb:rhp} has the same symmetry property. Thus for any $\lambda$ we have that
\begin{equation*}
\Xi^{(n)}(-\lambda)=\hat{\sigma} _{1}\Xi^{(n)}(\lambda)\hat{\sigma}_{1}.
\end{equation*}

Computing the asymptotic expansion at $ \infty $ of both sides of this equation, we have that
$ (-1)^{k}\Xi^{(n)} _{k}= \hat{\sigma} _{1}\Xi_{k}^{(n)}(\lambda)\hat{\sigma}_{1}$. This directly implies the two equations~\eqref{eq:symcoef} for $ k=2j $ or $ k=2j-1 $.

Concerning the statement for the asymptotic coefficients of the inverse of $ \Xi^{(n)} $, namely $ \Theta^{(n)} $, the proof follows by the fact that $ \Theta^{(n)} $ solves another Riemann--Hilbert problem, with same symmetry for the jump matrix. Indeed, consider the following problem for a function~$ \Theta^{(n)} $:
\begin{itemize}\itemsep=0pt\label{pb:invrhp}
\item $ \Theta^{(n)} $ is a $ (\lambda-) $analytic matrix-valued function on $\mathbb{C}\setminus \big( \gamma^{n}_{+} \cup \gamma^{n}_{-}\big) $ admitting continuous extension from either side to $\gamma^{n}_{+} \cup \gamma^{n}_{-}$;
\item it has a jump condition for each $ \lambda \in \gamma^{n}_{+} \cup \gamma^{n}_{-} $
\begin{equation*}
\Theta^{(n)}_{+}(\lambda)=
\underbrace{\begin{pmatrix}
I_{r} & 2\pi {\rm i} r^{(n)} (\lambda)\chi_{\gamma^{n}_{+}}(\lambda)\\
2\pi {\rm i} r^{(n)} (-\lambda)\chi_{\gamma^{n}_{-}}(\lambda)& I_{r}
\end{pmatrix}}_{:=H^{(n)}(\lambda,\vec{s})}\Theta_{-} ^{(n)}(\lambda);
\end{equation*}
\item it has the asymptotic condition for $|\lambda| \rightarrow \infty $
\begin{equation*}
\Theta^{(n)}(\lambda) \sim I_{2r } + \sum_{j \geq 1 } \frac{\Theta^{(n)}_{j}}{\lambda^{j}}.
\end{equation*}
\end{itemize}
The function $ \Theta^{(n)} $ with these properties is the inverse of the solution of Problem~\ref{pb:rhp}. Indeed: the functions $\Theta^{(n)}\Xi^{(n)}(\lambda)$, and $\Xi^{(n)}\Theta^{(n)} $ have no jumps along $ \gamma^{n}_{+} \cup \gamma^{n}_{-} $ and they both behave like the identity matrix at $ \infty $. Thus by the generalized Liouville theorem, they both have to coincide with the identity matrix.

We then observe that the jump matrix $H^{(n)} $ here has the same symmetry property of $ J^{(n)}$, i.e., $ \hat{\sigma}_{1}H^{(n)}(\lambda,\vec{s})\hat{\sigma}_{1}=H^{(n)}(-\lambda,\vec{s}), $ for each $ \lambda \in \gamma^{n}_{+} \cup \gamma^{n}_{-}$. Thus, exactly as before, even the function $ \Theta^{(n)} $ has the same property:
\begin{equation*}
\hat{\sigma}_{1}\Theta^{(n)}(\lambda,\vec{s})\hat{\sigma}_{1}=\Theta^{(n)}(-\lambda,\vec{s}).
\end{equation*}
We conclude then that the asymptotic coefficients of $ \Theta^{(n)} $ have the same form of the $ \Xi_{k}$, i.e.,
\begin{gather}
\Theta^{(n)}_{2j}=\tilde{\alpha}_{2j}^{(n)} \otimes I_{2r}+\tilde{\beta}_{2j}^{(n)}\otimes\sigma_{1},\nonumber\\
\Theta^{(n)}_{2j-1}=\tilde{\alpha}_{2j-1}^{(n)} \otimes \sigma_{3}+\tilde{\beta}_{2j-1}^{(n)}\otimes\sigma_{2}, \qquad j\geq 1,\label{eq:symcoefinv}
\end{gather}
where, as before, $\tilde{\alpha}_{l}^{(n)}$ and $ \tilde{\beta}_{l}^{(n)}$ for every $l\geq 1$ correspond to the $r\times r$ matrices in the entries~$(1,1)$ and $(1,2)$ of the block matrix $\Theta_l^{(n)}$.
\end{proof}

We are now ready to state the fundamental result that connects the matrix Airy Hankel operators to these Riemann--Hilbert problems.

Supposing that the solutions of the Riemann--Hilbert Problem~\ref{pb:rhp} and its inverse exist, we have the following result.
\begin{Remark}
Existence conditions for $ \Xi^{(n)} $ (and thus $\Theta^{(n)}$) are given at the end of the section (see Theorem~\ref{Theorem:existence solution rhp}).
\end{Remark}

\begin{Theorem}\label{Proposition:ascoeffredolm}
For each $ n \in \mathbb{N} $, consider $ \Xi^{(n)} $ the solution of the Riemann--Hilbert Problem~{\rm \ref{pb:rhp}} and its inverse $\Theta^{(n)}:= \big( \Xi^{(n)} \big)^{-1} $. Then the following identities hold
\begin{gather}
\dS \ln \big( F^{(n)}(s_{1},\dots,s_{r})\big) =\int_{\gamma^{n}_{+}\cup\gamma^{n}_{-} }\Tr\left(\Theta_{-}^{(n)} \big( \Xi_{-}^{(n)} \big) '\dS J^{(n)} \big( J^{(n)}\big) ^{-1} \right) \frac{{\rm d}\lambda}{2\pi {\rm i}}\nonumber\\
\hphantom{\dS \ln \big( F^{(n)}(s_{1},\dots,s_{r})\big)}{} =-2{\rm i} \Tr\big( \alpha^{(n)}_{1}\big),\label{eq:alphafred}
\end{gather}
where in the integral in the middle we indicate with $ ' $ the derivation w.r.t.\ the spectral parame\-ter~$ \lambda $ and the differential operator~$ \dS $ is defined as in~\eqref{eq:ds}.
\end{Theorem}
\begin{proof}The proof follows as an application to this very specific case of some general result obtained in~\cite{bertola2012fredholm}. We split the proof in two parts, one for each equality in~\eqref{eq:alphafred}.

In order to obtain the first equality we need essentially two results. The first one establishes the relation between Fredholm determinant of the Airy matrix operator and Fredholm determinant of certain integral kernel operator \cite[Corollary~2.1]{bertola2012fredholm}). In particular, we first get that the
Fredholm determinants of $\lbrace \mathcal{A}{\rm i}_{2n+1}\rbrace_{n \in \mathbb{N}} $ are equal to the ones of the integral operators acting on $ L^{2}\big(\gamma_{+}^{(n)},\mathbb{C}^{r} \big) $ with kernels
\begin{equation} \label{eq:intkern}
\mathcal{K}^{(n)}( \lambda, \mu) =\frac{r^{(n)}(\lambda, \mu ) }{\lambda+\mu},
\end{equation}
with $ r^{(n)}(\lambda, \mu) $ defined as in~\eqref{eq:funcrjump}.

As by product we then have that
\begin{equation*}
F^{(n)}(s_{1},\dots,s_{r})=\det \big({\rm Id}_{\gamma_{+}^{(n)}}- \big( \mathcal{K}^{(n)}\big)^{2} \big) .
\end{equation*}
The second result needed comes from the study of matrix integral kernels of type~\eqref{eq:intkern}, through Riemann--Hilbert problems. Indeed, it allows to compute the Fredholm determinants of these integrable operators in terms of the solution of Riemann--Hilbert Problem~\ref{pb:rhp}. In particular, by applying \cite[Theorem~4.1]{bertola2012fredholm}, we have that
\begin{equation*}
\int_{\gamma^{n}_{+}\cup\gamma^{n}_{-} }\Tr\left(\Theta_{-}^{(n)} \big( \Xi_{-}^{(n)} \big) '\dS J^{(n)} \big( J^{(n)}\big) ^{-1} \right) \frac{{\rm d}\lambda}{2\pi {\rm i}}= \dS \ln \det\big( I_{\gamma_{+}^{(n)}}-\big( \mathcal{K}^{(n)}\big) ^{2}\big) .
\end{equation*}
Thus the first identity in the statement holds.

For what concerns the second identity of the statement, we proceed by direct computation of the integral
\begin{equation}
\label{eq:integral}
\int_{\gamma^{n}_{+}\cup\gamma^{n}_{-} }\Tr\big(\Theta_{-}^{(n)} \big( \Xi_{-}^{(n)} \big) '\dS J^{(n)} \big( J^{(n)}\big) ^{-1} \big) \frac{{\rm d}\lambda}{2\pi {\rm i}}.
\end{equation}
First of all, we observe that the jump matrix $ J^{(n)}(\lambda,\vec{s}) $ that appears in the jump condition~\eqref{eq:jump}, admits the factorization
\[ J^{(n)}(\lambda,\vec{s})= \exp\big( \theta^{(n)}(\lambda,\vec{s})\otimes \sigma_{3} \big)J^{(n)}_{0}\exp\big({-} \theta^{(n)}(\lambda,\vec{s})\otimes \sigma_{3} \big), \]
with $ J^{(n)}_{0} $ the constant matrix given by
\[J^{(n)}_{0}= \begin{pmatrix}
I_{r}&C\\
C&I_{r}
\end{pmatrix}.\]
Thus we can easily compute the second factor appearing under the trace in the integral~\eqref{eq:integral}:
\begin{equation}\label{eq:factJ}
\left( \dS J^{(n)} \right) \big( J^{(n)}\big) ^{-1}= {\rm i}\lambda\hat{\sigma}_{3}-J^{(n)}\big({\rm i}\lambda \hat{\sigma}_{3} \big)\big( J^{(n)}\big) ^{-1} .
\end{equation}
We are now going to show that the integral in~\eqref{eq:integral} is actually just the formal residue at~$ \infty $ of a certain function. Furthermore in this particular case, due to the form of the matrix~$J^{(n)}$, the residue can be explicitly computed using equation~\eqref{eq:factJ}.

To start with, we consider the following function
\begin{equation}
\label{eq:ausfun} \Tr\left(\Theta^{(n)} \big( \Xi^{(n)}\big)' \dS\big( \theta^{(n)} \otimes \sigma_{3}\big) \right)= \Tr\big(\Theta^{(n)} \big( \Xi^{(n)}\big)' {\rm i}\lambda\hat{\sigma} _{3}\big) .
\end{equation}
Its formal residue at $ \infty $ can be computed as
\[ -\Res_{\lambda=\infty}\Tr\big(\Theta^{(n)} \big( \Xi^{(n)}\big)' {\rm i}\lambda\hat{\sigma_{3}}\big) =\lim_{R\rightarrow \infty}\int_{|\lambda|=R} \Tr\big(\Theta^{(n)} \big( \Xi^{(n)}\big) ' {\rm i}\lambda\hat{\sigma_{3}} \big)\frac{{\rm d}\lambda}{2\pi {\rm i}}. \]
Now, this counterclockwise circle for $ R \rightarrow\infty $, can be deformed like $ \gamma_{+}^{(n)} \cup\gamma_{-}^{(n)}$. As byproduct, the formal residue of \eqref{eq:ausfun} can be rewritten, taking into account the boundary values of $ \Theta^{(n)} $ and $ \big( \Xi^{(n)}\big) '$ along the curves $ \gamma_{\pm}^{(n)} $, as follows
\[ \int_{\gamma^{n}_{+}\cup\gamma^{n}_{-} }\Tr\big( \big({-}\Theta_{+}^{(n)} \big( \Xi_{+}^{(n)}\big)'+\Theta_{-}^{(n)} \big( \Xi_{-}^{(n)}\big)'\big) {\rm i}\lambda\hat{\sigma_{3}}\big)\frac{{\rm d}\lambda}{2\pi {\rm i}} . \]
 Now, from the jump condition \eqref{eq:jump} we deduce that all along the curves $ \gamma_{\pm}^{(n)} $ we have the relation
\[ \big( \Xi_{+}^{(n)}\big) '= \big( \Xi_{-}^{(n)}\big) ' J^{(n)} + \big( \Xi_{-}^{(n)}\big) \big( J^{(n)}\big) '. \]
Thus replacing it in the first integral above we get
\begin{gather*}
\int_{\gamma^{n}_{+}\cup\gamma^{n}_{-} }\Tr\big( \big({-}\Theta_{+}^{(n)} \big( \Xi_{+}^{(n)}\big)'+\Theta_{-}^{(n)} \big( \Xi_{-}^{(n)}\big)'\big) {\rm i}\lambda\hat{\sigma_{3}}\big)\frac{{\rm d}\lambda}{2\pi {\rm i}}\\
\qquad{}=-\int_{\gamma^{n}_{+}\cup\gamma^{n}_{-} }\Tr\big(\big( \big( J^{(n)}\big) ^{-1}\Theta^{(n)}_{-}\big(\big( \Xi_{-}^{(n)}\big)'J^{(n)}+\Xi_{-}^{(n)}\big( J^{(n)}\big) ^{'} \big) -\Theta_{-}^{(n)}\big( \Xi_{-}^{(n)}\big)' {\rm i}\lambda \hat{\sigma}_{3} \big) \big)\frac{{\rm d}\lambda}{2\pi{\rm i}} \\
\qquad{}=-\int_{\gamma^{n}_{+}\cup\gamma^{n}_{-} }\Tr\big( \big( \big( J^{(n)}\big) ^{-1}\Theta^{(n)}_{-}\big( \Xi_{-}^{(n)}\big)' J^{(n)}+\big( J^{(n)}\big)^{-1} J '-\Theta_{-}^{(n)}\big( \Xi_{-}^{(n)}\big)'\big){\rm i}\lambda \hat{\sigma}_{3}\big)\frac{{\rm d}\lambda}{2\pi{\rm i}}\\
\qquad{}=-\int_{\gamma^{n}_{+}\cup\gamma^{n}_{-} }\Tr\big(\Theta^{(n)}_{-}\big( \Xi_{-}^{(n)}\big)'\big( J^{(n)}{\rm i}\lambda\hat{\sigma}_{3}\big( J^{(n)}\big) ^{-1}-{\rm i}\lambda\hat{\sigma}_{3}\big)\big) \frac{{\rm d}\lambda}{2\pi{\rm i}}\\
\qquad{}=\int_{\gamma^{n}_{+}\cup\gamma^{n}_{-} }\Tr\left(\Theta^{(n)}_{-}\big( \Xi_{-}^{(n)}\big)'\dS J^{(n)}\big( J^{(n)}\big) ^{-1}\right)\frac{{\rm d}\lambda}{2\pi{\rm i}},
\end{gather*}
where in the last passages we used the invariance of the trace by conjugation and the fact that the quantity $ \big( J^{(n)}\big)^{-1} \big(J^{(n)}\big)^{'} {\rm i}\lambda\hat{\sigma}_{3} $ is trace free.

Finally, using the asymptotic expansion at $ \infty $ given in \eqref{eq:asymp}, we get that
\[ \Res_{\lambda=\infty}\Tr\big(\Theta^{(n)} \big( \Xi^{(n)}\big)' {\rm i}\lambda\hat{\sigma}_{3}\big)=-2{\rm i}\Tr\big( \alpha_{1}^{(n)}\big) ,\]
and this concludes the proof.
\end{proof}

\begin{Remark}
In the study of isomonodromy deformations, the quantity
\begin{equation*}
\int_{\gamma^{n}_{+}\cup\gamma^{n}_{-} }\Tr\left(\Theta_{-} \Xi_{-} '\dS J^{(n)} \big( J^{(n)}\big) ^{-1} \right) \frac{{\rm d}\lambda}{2\pi{\rm i}}
\end{equation*}
is associated to the isomonodromic tau function $\tau_{\Xi^{(n)}}$ related to the Riemann--Hilbert Problem~\ref{pb:rhp} depending on the parameters $ \lbrace s_{k} \rbrace_{k=1}^{r} $, through the formula
\begin{equation*}
\dS \ln \tau_{\Xi^{(n)}}=\int_{\gamma^{n}_{+}\cup\gamma^{n}_{-} }\Tr\left(\Theta_{-} \Xi_{-} '\dS J^{(n)} \big( J^{(n)}\big) ^{-1} \right) \frac{{\rm d}\lambda}{2\pi{\rm i}}.
\end{equation*}
This notion was first introduced in~\cite{jimbo1981monodromy}, and then generalized for example in~\cite{bertola2010dependence}. With Theorem~\ref{Proposition:ascoeffredolm} we recover for any Airy matrix Hankel operator~\eqref{eq:convop} the relation between the Fredholm determinant $ F^{(n)}(s_{1},\dots,s_{r})$ and the isomonodromic tau function associated to the Riemann--Hilbert Problem~\ref{pb:rhp}, that was proved in Theorem~4.1 of~\cite{bertola2012fredholm} for Fredholm determinants of generic matrix Hankel operators.
\end{Remark}

Finally, in order to use the formula \eqref{eq:alphafred} for the logarithmic derivative of $ F^{(n)}(s_{1},\dots,s_{r}) $, we need to find out whether the solution $\Xi^{(n)} $ of the Riemann--Hilbert Problem~\ref{pb:rhp} exists or not. In particular, we are going to see that under certain assumptions on the constant matrix~$ C $, the existence of~$\Xi^{(n)} $ is assured. The following result is indeed a generalization of Theorem~5.1 in~\cite{bertola2012fredholm}, for the generalized Airy matrix operators defined in~\eqref{eq:convop}.
\begin{Theorem}\label{Theorem:existence solution rhp}
Let the matrix~$ C $ be Hermitian, then the solution~$\Xi^{(n)} $ of the Riemann--Hilbert Problem~{\rm \ref{pb:rhp}} exists if and only if the eigenvalues of~$ C $ lay in the interval $[-1,1] . $
\end{Theorem}

Before starting the proof of Theorem~\ref{Theorem:existence solution rhp}, we
state the following lemma. For $n=1$ the result is known from~\cite{basor1999determinants,hastings1980boundary}. In the following we adapted the proof to the case of generic~$n$. For finite $z\in \mathbb{R}$, we introduce the operator
\begin{equation*}
\big( \Phi_{\mathrm{Ai}_{2n+1}}^z f\big) (x) = \int_{z}^{+\infty} \mathrm{Ai}_{2n+1} (x+y)f(y)\,{\rm d}y,\qquad f \in L^2(\mathbb{R}).
\end{equation*}
\begin{Lemma}\label{Lemma:unitary airy}
For any $ n \in \mathbb{N} $ we consider the Airy transform $ \Phi_{\mathrm{Ai}_{2n+1}}$ acting on $f \in L^2(\mathbb{R}) \cap L^1(\mathbb{R}) $ as
\begin{equation}
\label{eq:airy transform}
( \Phi_{\mathrm{Ai}_{2n+1}}f) (x)=\lim_{z\rightarrow -\infty}\big( \Phi_{\mathrm{Ai}_{2n+1}}^zf\big) (x) = \lim_{z\rightarrow -\infty}\left(\int_{z}^{+\infty} \mathrm{Ai}_{2n+1} (x+y)f(y)\,{\rm d}y \right) .
\end{equation}
Then $\lim\limits_{z\rightarrow -\infty}\big| \big| \Phi_{\mathrm{Ai}_{2n+1}}^z f\big| \big| = | | f | | $ for the $L^2(\mathbb{R})$-norm, and thus for any finite $ z$ the inequality $\big| \big| \big|\Phi_{\mathrm{Ai}_{2n+1}}^z\big| \big| \big| \leq 1 $ holds for the $L^2((z,+\infty))$ operator norm.
\end{Lemma}
\begin{proof}
We consider $ \Phi_{\mathrm{Ai}_{2n+1}} $ the Airy transform acting as defined in~\eqref{eq:airy transform},
where inside the integral we have the scalar Airy function $\mathrm{Ai}_{2n+1}$ defined in~\eqref{eq:int representation airy}, without any shift and for real values of~$ x $.
 We introduce the Fourier transform $ \mathfrak{F} $ and its inverse $ \mathfrak{F}^{-1} $ defined on $L^2(\mathbb{R}) \cap L^1(\mathbb{R})$ (and extended to $L^{2}(\mathbb{R})$ by continuity and density argument), in the standard way as
\begin{equation*}
 ( \mathfrak{F}h ) (x):=\frac{1}{\sqrt{2\pi}}\int_{\mathbb{R}}h(\lambda)\exp (-{\rm i}x\lambda ) \, {\rm d}\lambda, \qquad
\mathfrak{F}^{-1} := \mathfrak{F} \mathfrak{I} = \mathfrak{I} \mathfrak{F},
\end{equation*}
 where $( \mathfrak{I}h ) (x)=h(-x)$, and the multiplication operator by $ \exp\big(\frac{{\rm i}x^{2n+1}}{2n+1}\big) $, denoted by $ \mathcal{M}_{n}. $
Then we observe that the Airy transform $ \Phi_{\mathrm{Ai}_{2n+1}}$ can be rewritten as the composition of these operators, in such a way that
\begin{equation*}
\Phi_{\mathrm{Ai}_{2n+1}}=\mathfrak{F}^{-1} \mathcal{M}_{n}\mathfrak{F}^{-1} = \mathfrak{F} \mathfrak{I} \mathcal{M}_{n} \mathfrak{I} \mathfrak{F}=\mathfrak{F} \mathcal{M}_{n}^{-1} \mathfrak{F} = \Phi_{\mathrm{Ai}_{2n+1}}^{-1}.
\end{equation*}
This implies that
\begin{align}
\lim_{z\rightarrow -\infty} \big| \big| \Phi_{\mathrm{Ai}_{2n+1}}^z f\big| \big| & = \lim_{z\rightarrow -\infty}\left( \int_{\mathbb{R}} \big| \Phi_{\mathrm{Ai}_{2n+1}}^z f (y)\big|^2 {\rm d}y\right) ^{\frac{1}{2}} \nonumber\\
& =\left( \int_{\mathbb{R}} \left| \int_{\mathbb{R}} \mathrm{Ai}_{2n+1}(y+u) f (u){\rm d}u \right|^2 {\rm d}y\right) ^{\frac{1}{2}}= | | f | | ,\label{eq:limit}
\end{align}
the norms being in $L^2(\mathbb{R})$.

Now we prove by contradiction the last inequality $\big| \big| \big|\Phi_{\mathrm{Ai}_{2n+1}}^z\big|\big|\big| \leq 1 $ for the $L^2((z,+\infty))$ operator norm. Suppose that there exist a scalar $\mu$ and an eigenfunction $g^z\in L^2((z,+\infty))$ such that $\Phi_{\mathrm{Ai}_{2n+1}}^z g^z =\mu g^z$ and $ |\mu | >1$. Then we can define $g\in L^2(\mathbb{R})$ as
\begin{equation*}
g(y) = \begin{cases}
 g^z (y), & \text{for}\ y\geq z,\\
 0, & \text{for} \ y< z,
\end{cases}
\end{equation*}
and we obtain for $\tilde{z} \leq z$ that $\Phi_{\mathrm{Ai}_{2n+1}}^{\tilde{ z}}g (y) =\Phi_{\mathrm{Ai}_{2n+1}}^z g^z (y) = \mu g^z(y)=\mu g(y)$ for $y\geq z$. Finally, since $ | \mu |>1 $, we have
\begin{equation*}
\big|\big| \Phi_{\mathrm{Ai}_{2n+1}}^{\tilde{ z}}g\big| \big| _{L^2(\mathbb{R})} \geq \big| \big| \Phi_{\mathrm{Ai}_{2n+1}}^{\tilde{ z}}g \big| \big| _{L^2((z,+\infty))} = | \mu | | | g | | _{L^2((z,+\infty))}> | | g | | _{L^2(\mathbb{R})}
\end{equation*}and this is in contradiction with equation \eqref{eq:limit}.
\end{proof}

We can finally provide a complete proof of Theorem \ref{Theorem:existence solution rhp}.
\begin{proof}By applying Theorem 3.1 of \cite{bertola2012fredholm} (which generalizes the fundamental result obtained first in \cite{its1990differential}) to the sequence of operators $\big( \mathcal{K}^{(n)}\big)^{2}, $ we have that the solutions $\Xi^{(n)} $ of the Riemann--Hilbert Problem~\ref{pb:rhp} exist if and only if the operator ${\rm Id} - \big( \mathcal{K}^{(n)}\big)^{2} $ is invertible. This is guaranteed by the non vanishing condition of the quantity $ \det\big({\rm Id}- \big( \mathcal{K}^{(n)}\big) ^{2}\big)=\det\big( {\rm Id}- \mathcal{A}{\rm i}_{2n+1}^{2}\big)$ (the equality follows as before from Corollary~2.1 of~\cite{bertola2012fredholm}) that is verified if the operators $\mathcal{A}{\rm i}_{2n+1}$ are such that $ | | |\mathcal{A}{\rm i}_{2n+1} | | | < 1. $ Here and in the following, $ | | | \cdot | | |$ stands for the operator norm induced from the $L^2$-norms on the domain and codomain of the relevant operator.

Supposing that the eigenvalues of $C $ are in the interval $[-1,1] , $ we are going to show that the inequality for the operator norm of $ \mathcal{A}{\rm i}_{2n+1} $ holds. Since the operators $ \mathcal{A}{\rm i}_{2n+1} $ defined in~\eqref{eq:convop}, are constructed by shifting by some component of $ \vec{s} $ the Airy function, we first observe that:
\begin{equation*}
 | | |\mathcal{A}{\rm i}_{2n+1} | | | = \big| \big| \big|\mathrm{P}_{s}\mathcal{A}{\rm i}_{2n+1}^{\vec{0}} \mathrm{P}_{s} \big| \big|\big|,
\end{equation*}
where $ \mathcal{A}{\rm i}_{2n+1}^{\vec{0}} $ is the operator without any shift, namely
\begin{equation*}
\mathcal{A}{\rm i}_{2n+1}^{\vec{0}}\mathbf{f} (x):=\int_{\mathbb{R}_{+}} \mathbf{Ai}_{2n+1}(x+y,\vec{0})\mathbf{f}(y)\,{\rm d}y .
\end{equation*}
considered from and to the space $\bigoplus_{k=1}^{r} L^2 ( [ s_{k},+\infty ) , \mathbb{C} )$
 and $ \mathrm{P}_{s} $ is the orthogonal projection
\begin{equation*}
\mathrm{P}_{s} \colon \ L^2\big( \mathbb{R}, \mathbb{C}^r \big) \longrightarrow \bigoplus_{k=1}^{r} L^2 ( [ s_{k},+\infty ) , \mathbb{C} )
\end{equation*}
acting diagonally as $\mathrm{P}_{s} :=\diag ( \chi_{[ s_{k},+\infty)}) _{k=1}^{r}$.
From equation \eqref{eq:comp conv opera}, we can see the matrix operators $ \mathcal{A}{\rm i}_{2n+1} $ written in terms of the scalar operators $ \Phi_{\mathrm{Ai}_{2n+1}}^z $ through tensor product. In particular, when there is no shift we simply have
\begin{equation*}
\mathcal{A}{\rm i}_{2n+1}^{\vec{0}}=C\otimes\Phi_{\mathrm{Ai}_{2n+1}}^0.
\end{equation*}
Finally, using the property of the scalar operator $ \Phi_{\mathrm{Ai}_{2n+1}}^z $ proved in Lemma~\ref{Lemma:unitary airy}, we conclude that
\begin{equation*}
\big|\big|\big|\mathcal{A}{\rm i}_{2n+1}^{\vec{0}}\big|\big|\big|= | | |C | | | \big| \big| \big|\Phi_{\mathrm{Ai}_{2n+1}}^0 \big|\big|\big| \leq | | |C | | | ,
\end{equation*}
where the matrix norm of $C$ above is induced by the $2$-norm on $\mathbb{C}^r$, i.e., it corresponds to the spectral radius of $C$. Then we have
\begin{equation*}
|||\mathcal{A}{\rm i}_{2n+1}|||\leq |||\mathrm{P}_{s}|||\big|\big|\big|\mathcal{A}{\rm i}_{2n+1}^{\vec{0}}\big|\big|\big| |||\mathrm{P}_{s}|||< |||C||| \leq 1,
\end{equation*}
and this concludes the proof of one of the implications in the statement.

In order to prove the other implication, we suppose that there exist $ \lambda_{0} $ eigenvalue of $ C $ such that $|\lambda_{0} | >1 $, with corresponding eigenvector $ \mathbf{v}_{0} \in \mathbb{C}^{r}$. In this case, we will be able to construct a nonzero function $ \mathbf{f}_{s}(x) $ such that there exist a value $ s_{0} $ for which
\begin{equation*}
\mathcal{A}{\rm i}_{2n+1}^{2}\mathbf{f}_{s_0}(x)=\mathbf{f}_{s_{0}}(x),
\end{equation*}
so we have that the operator $ {\rm Id}- \mathcal{A}{\rm i}_{2n+1} ^{2} $ is not invertible and thus the solution of the Riemann--Hilbert Problem~\ref{pb:rhp} does not exist.

Indeed, consider $ \mathbf{f}(x):= \mathbf{v}_{0}f(x) $, for any scalar function $ f \in L^2(\mathbb{R}) $. Then applying the operator $ \mathcal{A}{\rm i}_{2n+1}^{2} $ with a shift $ \vec{s}=(s,\dots,s) $ for a certain $s \in \mathbb{R}$ we have
\begin{equation*}
\mathcal{A}{\rm i}_{2n+1}^{2}\mathbf{f}(x)=\lambda_{0}^{2}\mathbf{v}_{0}\int_{\mathbb{R}_{+}}K_{\mathrm{Ai}_{2n+1}}(x+s,y+s)f(y)\, {\rm d}y,
\end{equation*}
where $ K_{\mathrm{Ai}_{2n+1}}$ is the $ n$-th generalized scalar Airy kernel. The corresponding kernel operator is self-adjoint, trace-class and in particular compact, acting on $ L^2([s,\infty)) $. We consider its maximum eigenvalue $ \mu(s) $ and the corresponding eigenfunction $ f_{s}(x) $. Finally by taking $ \mathbf{f}_{s}(x)=\mathbf{v}_{0}f_{s}(x) $ we get
\begin{equation*}
\mathcal{A}{\rm i}_{2n+1}^{2}\mathbf{f}_{s}(x)=\lambda_{0}^{2}\mu(s)\mathbf{f}_{s}(x).
\end{equation*}
Since $ \lambda_{0}^{2}>1 $ and $ \mu(s) $ is a continuous function such that $ \mu(s) \rightarrow 1 $ for $ s\rightarrow -\infty $ and $ \mu(s) \rightarrow 0 $ for $ s\rightarrow +\infty $, there exist a value $ s_{0} \in \mathbb{R} $ for which the above equation reads as
\begin{equation*}
\mathcal{A}{\rm i}_{2n+1}^{2}\mathbf{f}_{s_{0}}(x,\vec{s_{0}})=\mathbf{f}_{s_{0}}(x).
\end{equation*}
And this completes the proof.
\end{proof}
\begin{Remark}
As byproduct of the theorem above, we have that the operators $ \mathcal{A}{\rm i}^{2}_{2n+1}$ are bounded from above by the identity.
We can actually show that any of the operators $ \mathcal{A}{\rm i}^{2}_{2n+1}$ is also limited from below: indeed they are all totally positive on $ \mathcal{C}:= \lbrace 1,\dots,r \rbrace \times \mathbb{R} $.
The main idea to show this is to interpret $ \mathcal{A}{\rm i}^{2}_{2n+1}$ as a scalar function on $ \mathcal{C}\times \mathcal{C}, $ in this way: for any couple $ (\xi_{1},\xi_{2}) =((j_{1},x_{1}), (j_{2},x_{2}))\in \mathcal{C}\times \mathcal{C} $ we have
\begin{equation*}
\mathcal{A}{\rm i}^{2}_{2n+1}(\xi_{1},\xi_{2}) = \sum_{ k=1}^{r} c_{j_{1},k}c_{k, j_{2}} \int_{\mathbb{R}_{+}} \mathrm{Ai}_{2n+1}(x_{1}+z+s_{j_{1}}+s_{k})\mathrm{Ai}_{2n+1}(x_{2}+z+s_{j_{2}}+s_{k})\, {\rm d}z.
\end{equation*}
In this way the claim is proved if we prove that for any natural $ L $, the quantity
\[\det\big( \mathcal{A}{\rm i}^{2}_{2n+1} (\xi_{a},\xi_{b})\big)_{a,b \leq L}\] is positive.

In order to do this, we first rewrite $ \mathcal{A}{\rm i}^{2}_{2n+1}(\xi_{1},\xi_{2}) $ using the product measure $ d\mu(\xi) $ on $ \mathcal{C} $ given by the product of the counting measure on $\left\lbrace 1,\dots,r\right\rbrace$ and the Lebesgue measure on $ \mathbb{R}. $ Thus
\begin{equation}
\mathcal{A}{\rm i}^{2}_{2n+1} (\xi_{a},\xi_{b})=\int_{\mathcal{C}_{+}}F_{2n+1}(\xi_{a},\zeta)F_{2n+1}(\zeta,\xi_{b})\,{\rm d}\mu(\xi),
\end{equation}
where we defined the function $F_{2n+1}(\xi_{a},\zeta)=c_{j_{a},k}\mathrm{Ai}_{2n+1}(x_{1}+z+s_{j_{a}}+s_{k}) $.
In this way we can determine the sign of the determinant, indeed
\begin{align*}
\det\big( \mathcal{A}{\rm i}^{2}_{2n+1} (\xi_{a},\xi_{b})\big)_{a,b \leq L} &= \det \left( \int_{\mathcal{C}_{+}}F_{2n+1}(\xi_{a},\zeta)F_{2n+1}(\zeta,\xi_{b})\,{\rm d}\mu(\xi)\right) _{a,b \leq L} \\
&=\frac{1}{L!}\int_{\mathcal{C}^L_{+}}\det ( F_{2n+1}(\xi_{a},\xi_{c}) )\det ( F_{2n+1}(\xi_{c},\xi_{a}) )\prod_{c=1}^{L}{\rm d}\mu(\xi_{c})\\
&= \frac{1}{L!}\int_{\mathcal{C}^L_{+}} | \det ( F_{2n+1}(\xi_{a},\xi_{c}) ) |^2 \prod_{c=1}^{L}{\rm d}\mu(\xi_{c}) >0,
\end{align*}
where in the first passage we used a general property in measure theory, the Andreief identity (see here \cite{baik2016combinatorics} for details), and in the last one we used the fact that~$ C $ is hermitian.
\end{Remark}

In conclusion, by taking $ C $ an hermitian matrix with eigenvalues laying in the interval $[-1,1] $, any $ \mathcal{A}{\rm i}^{2}_{2n+1}$ is hermitian and thanks to Theorem~\ref{Theorem:existence solution rhp} and the previous remark, we can say that any $ \mathcal{A}{\rm i}^{2}_{2n+1}$ defines a determinantal point processes on that space of configuration $ \mathcal{C} $ (directly by applying Theorem~3 of~\cite{soshnikov2000determinantal}). In particular this implies that the Fredholm determinants $ F^{(n)}(s_1,\dots,s_r)$ are the joint probability of the last points for some multi-process on $ \mathbb{R}$ (see for instance Proposition~2.9 of~\cite{johansson2005random}), namely
\begin{equation*}
F^{(n)}(s_1,\dots,s_r) =\mathbb{P}\big( x_{i}^{\text{max}} < s_{i},\, i=1,\dots,r\big).
\end{equation*}

\section{Matrix Painlev\'e II hierarchy}\label{sec:III}
In this section, we are finally going to define our noncommutative Painlev\'e II hierarchy.
In the following, we will consider $U(\vec{s})$, $W(\vec{s}) $ as functions depending on the parameters $ s_{1},\dots,s_{r} $ with values in ${\rm Mat} ( r\times r, \mathbb{C} ) $.

In this context we will use the standard notation for the commutator and anticommutator between two matrices:
\[ [ A,\cdot ]=A\cdot-\cdot A \qquad \text{and} \qquad [A,\cdot ]_{+}=A\cdot+\cdot A. \]
In order to define a fully noncommutative version of the PII hierarchy, we first define a sequence of differential polynomials $ \mathcal{L}_{n} \ubr $ through a matrix version of the Lenard operators. Following~\cite{gordoa2016matrix}:
\begin{gather}
\mathcal{L}_{0}\ubr= \frac{1}{2}I_{r},\nonumber \\
\dS \mathcal{L}_{n}\ubr=\left( \dSSS + [ U, \cdot ] _+ \dS +\dS [ U,\cdot ] _+ + [ U,\cdot ] \dinv [ U,\cdot ] \right) \mathcal{L}_{n-1}\ubr , \qquad n \geq 1.\label{eq:lenmatrixrecursion}
\end{gather}
Here $ I_r $ denotes the identity matrix, $\dS$ denotes the differential operator defined in \eqref{eq:ds} and $ \dinv $ denotes the corresponding formal antiderivative. The locality of these operators computed in $U$ follows from Theorem~6.2 in~\cite{olver2000classification}.
\begin{Example}
The first of the differential polynomials in $U$ given by the recursive formula~\eqref{eq:lenmatrixrecursion} read as follows:
\begin{gather*}
\mathcal{L}_{1}\ubr = U, \\
\mathcal{L}_{2} \ubr =U_{2S}+3U^{2},\\
\mathcal{L}_{3}\ubr = U_{4S}+ 5 [U , U_{2S} ]_{+}+5U_{S} ^{2}+10U^{3}.
\end{gather*}
From $ n \geq 3 $ the ``noncommutative'' character of these operators appears in form of anticommutators.
\end{Example}
\begin{Remark}
In the example above and in the following we use the shorter notation $\big( \dS\big)^{n} U=U_{nS}$ for any $ n \in \mathbb{N} $.
\end{Remark}

\begin{Definition}We define a matrix PII hierarchy as follows
\begin{equation}\label{eq:piinchier}
{\rm PII}^{(n)}_{\rm NC}\left[ \alpha_{n}\right]\colon \ \left( \dS + \acomw \right) \mathcal{L}_{n}\ubr =(-1)^{n+1}4^{n} [S, W ] _{+}+a_{n}I_r,
\end{equation}
where $ U $ is as in the scalar case, the Miura transform of the function~$ W $, i.e., $ U:=\dS W-W^{2}, $ and $a_n$ are scalar constants.
\end{Definition}

In particular we will study the homogeneous hierarchy, setting $ a_{n}=0 $ for each $ n $.
\begin{Remark}
It is also possible to define a more general hierarchy, in the following way
\begin{gather*}
{\rm PII}^{(n)}_{\rm NC} [ \alpha_{n} ]\colon \ \left( \dS + \acomw \right) \mathcal{L}_{n}\ubr+\sum_{l=1}^{n-1}t_{l} \left( \dS + \acomw \right) \mathcal{L}_{l}\ubr\nonumber\\
\hphantom{{\rm PII}^{(n)}_{\rm NC} [ \alpha_{n} ]\colon}{} \ \qquad{} =(-1)^{n+1}4^{n} [S, W ] _{+}+a_{n}I_r,
\end{gather*}
for some scalars $ t_{1},\dots, t_{n-1} $. We recover the hierarchy \eqref{eq:piinchier} setting up these scalars to $ 0$.
Another matrix hierarchy was introduced in~\cite{gordoa2016matrix}, but there the time variable is a scalar.
\end{Remark}

\begin{Example}
Here are the first three equations of the homogeneous hierarchy \eqref{eq:piinchier}.
\begin{itemize}\itemsep=0pt
\item For $ n=1 $ we obtain the noncommutative analogue of the homogeneous PII equation:
\begin{equation}\label{eq:ncpii}
{\rm PII}_{\rm NC}\colon \ W_{2S}=2W^{3}+4 [S, W ] _{+}.
\end{equation}
This coincides with the homogeneus version of the fully noncommutative PII equation studied in \cite{retakh2010noncommutative}, in a more general context of any noncommutative algebra with derivation.
\item For $ n=2 $ we have the $4$-th order equation:
\begin{gather*}
{\rm PII}^{(2)}_{\rm NC}\colon \ W_{4S}=6W^{5}+4\big[W^{2}, W_{2S}\big]_{+} +2WW_{2S}W+2\big[W_{S}^{2},W \big]_{+}\\
\hphantom{{\rm PII}^{(2)}_{\rm NC}\colon \ W_{4S}=}{} +6W_{S}WW_{S}-4^2 [S, W ] _{+}.
\end{gather*}
\item For $ n=3 $ we have the $6$-th order equation:
\begin{gather*}
{\rm PII}^{(3)}_{\rm NC}\colon \ W_{6S} = 20W^{7}-15\big[W_{2S},W^{4} \big]-20W^{2}W_{2S}W^{2}- 10\big[ WW_{2S}W,W^{2} \big] _{+}\\
\hphantom{{\rm PII}^{(3)}_{\rm NC}\colon \ W_{6S} =}{} -10\big[W_{S} ^{2},W^{3}\big] _{+}-15\big[WW_{S}^{2}W,W \big] _{+}-20W_{S}W^{3}W_{S}\\
\hphantom{{\rm PII}^{(3)}_{\rm NC}\colon \ W_{6S} =}{} -25\big[W_{S}WW_{S},W^{2}\big] _{+}-5\big[W_{S}W^{2}W_{S},W \big] _{+}-10WW_{S}WW_{S}W\\
\hphantom{{\rm PII}^{(3)}_{\rm NC}\colon \ W_{6S} =}{} +6\big[W_{4S} ,W^{2}\big]+2WW_{4S}W +4 (W_{S}W_{3S}W++WW_{3S}W_{S} )\\
\hphantom{{\rm PII}^{(3)}_{\rm NC}\colon \ W_{6S} =}{} +9 (WW_{S}W_{3S}+W_{3S}W_{S}W )+15 (W_{S}WW_{3S} +W_{3S}WW_{S} ) \\
\hphantom{{\rm PII}^{(3)}_{\rm NC}\colon \ W_{6S} =}{} +25\big[W_{2S},W_{S}^{2} \big] _{+}+20 W_{S}W_{2S}W_{S}\\
\hphantom{{\rm PII}^{(3)}_{\rm NC}\colon \ W_{6S} =}{} +11\big[ W_{2S}^{2},W\big]_{+}+20W_{2S}W W_{2S}+4^3 [S, W ] _{+}.
\end{gather*}
\end{itemize}
\end{Example}
A fundamental property of matrix Lenard operators (that we are going to use in the next section in order to find the Lax pair for the hierarchy \eqref{eq:piinchier}) is given by the following formula (see~\cite{gordoa2016matrix}).
\begin{Proposition}For each $ n \in \mathbb{N}$ the Lenard operator acting on the Miura transform factorizes like
\begin{equation}\label{eq:factlennc}
\dS \mathcal{L}_{n+1}\ubr = \left(\dS -\acomw \right)\left( \dS -\comw\dinv \comw \right) \left(\dS+\acomw \right) \mathcal{L}_{n}\ubr.
\end{equation}
\end{Proposition}

This formula is achieved by the direct computation of the recursive formula for the noncommutative Lenard operators computed in the Miura transform $ U=W_{S}-W^{2} $.

\section{The isomonodromic Lax pair}\label{sec:IV}
In this section we are finally going to find out a Lax pair for the noncommutative hierarchy~\eqref{eq:piinchier}, making use of the Riemann--Hilbert Problem~\ref{pb:rhp} introduced in Section~\ref{sec:II}. In this way we will also be able to show the relation between some solution of the hierarchy \eqref{eq:piinchier} and the Fredholm determinant of the matrix $n$-th Airy Hankel operator.

To start with, we consider a new sequence of functions, defined using the solution of the Riemann--Hilbert Problem~\ref{pb:rhp}.
\begin{Definition} For each $ n \in \mathbb{N} $, we construct
\begin{equation*}
\Psi^{(n)}(\lambda, \vec{s}):= \Xi^{(n)}(\lambda)\exp\big( \theta^{(n)}(\lambda) \otimes \sigma_{3} \big) .
\end{equation*}
\end{Definition}
It is easy to check that these functions $ \left\lbrace \Psi^{(n)}\right\rbrace_{n \in \mathbb{N}} $ actually solve a new sequence of Riemann--Hilbert problems, with constant jump conditions. Namely, the following problems.
\begin{pb}\label{pb:constjump}
Find a $ (\lambda-)$ analytic matrix-valued function \[\Psi^{(n)}(\lambda)\colon \ \mathbb{C}\setminus \big( \gamma^{n}_{+} \cup \gamma^{n}_{-}\big)\rightarrow {\rm GL}(2r, \mathbb{C}) \]
admitting continuous extension to the contour $ \gamma^{n}_{+} \cup \gamma^{n}_{-}$ from either side and such that it satisfies the following two conditions:
\begin{itemize}\itemsep=0pt
\item the jump condition for each $ \lambda \in \gamma^{n}_{+} \cup \gamma^{n}_{-} $
\begin{equation*}
\Psi^{(n)}_{+}(\lambda)=\Psi^{(n)}_{-}(\lambda)
\underbrace{\begin{pmatrix}
I_{r}&C\chi_{\gamma^{n}_{+}}(\lambda)\\
C\chi_{\gamma^{n}_{-}}(\lambda)&I_{r}
\end{pmatrix}}_{:=K^{(n)}};
\end{equation*}
\item the asymptotic condition for $|\lambda| \rightarrow \infty $
\begin{equation*}
\Psi^{(n)}(\lambda)\sim \left( I_{2r } + \sum_{j \geq 1 } \frac{\Xi^{(n)}_{j}}{\lambda^{j}} \right) \exp\big( \theta^{(n)}(\lambda) \otimes \sigma_{3} \big) .
\end{equation*}
\end{itemize}
\end{pb}

As it is standard in the theory of isomonodromic deformations, we deduce the Lax pair for the noncommutative PII hierarchy \eqref{eq:piinchier} from the Riemann--Hilbert problems with piecewise constant jumps solved by $\Psi^{(n)}$.
The main idea is the following: using the fact that each $ \Psi ^{(n)} $ has constant jump condition (i.e., the jump matrix $ K^{(n)} $ does not explicitly depend on the spectral parameter $ \lambda $ or the deformations parameters $ s_{i}$, $i=1,\dots,r $), we can thus conclude that the quantities
\begin{equation}\label{eq:logderpsi}
\dS\Psi^{(n)}\big(\Psi^{(n)} \big) ^{-1}=:L^{(n)} \qquad \text{and} \qquad \frac{\partial}{\partial \lambda} \Psi^{(n)}\big(\Psi^{(n)} \big) ^{-1}=:M^{(n)}
\end{equation}
are matrix-valued polynomials in $ \lambda$.

\begin{Remark}Here the inverse of $ \Psi^{(n)} $ is simply given by
\begin{equation*}
\big( \Psi^{(n)}\big) ^{-1}(\lambda)=\exp\big( {-} \theta^{(n)}(\lambda) \otimes \sigma_{3}\big) \Theta^{(n)}(\lambda).
\end{equation*}
\end{Remark}

Furthermore, by using the symmetries of the Riemann--Hilbert Problem~\ref{pb:rhp}, we can compute the exact form of the coefficients of these polynomials $ L^{(n)}$, $M^{(n)}$.

The final result is summarized in the proposition below.
\begin{Proposition}\sloppy
There exist two polynomial matrices in $ \lambda $, which we denote with $L^{(n)}$ and~$M^{(n)}$, respectively of degree~$ 1 $ and~$ 2n $, such that the following system of differential equations is satisfied:
\begin{gather}
\dS\Psi ^{(n)}(\lambda, \vec{s})= L^{(n)}(\lambda, \vec{s})\Psi ^{(n)} ( \lambda, \vec{s} ) ,\nonumber\\
\partial_{\lambda}\Psi ^{(n)}(\lambda, \vec{s})=M^{(n)}(\lambda, \vec{s})\Psi ^{(n)} ( \lambda, \vec{s} ).\label{eq:laxpairRH}
\end{gather}
Moreover, $ L^{(n)} $ and $ M^{(n)} $ have the following forms
\begin{equation*}
L^{(n)}(\lambda, \vec{s})= \begin{pmatrix}
{\rm i}\lambda I_{r}&W(\vec{s})\\
W(\vec{s})& - {\rm i}\lambda I_{r}
\end{pmatrix}, \qquad \text{with} \quad W(\vec{s})=2\beta_{1}^{(n)}(\vec{s}),
\end{equation*}
and
\begin{equation*}
M^{(n)}(\lambda, \vec{s})=\begin{pmatrix}
A(\lambda,\vec{s})+{\rm i}S&{\rm i}G(\lambda,\vec{s})\\
-{\rm i}G(\lambda,\vec{s})&- A(\lambda,\vec{s})-{\rm i}S
\end{pmatrix}+\begin{pmatrix}
E(\lambda,\vec{s})&F(\lambda,\vec{s})\\
F(\lambda,\vec{s})&E(\lambda,\vec{s})
\end{pmatrix},
\end{equation*}
where
\begin{gather*}
A(\lambda,\vec{s})=\sum_{ k=0}^{n}\frac{{\rm i}}{2} \lambda ^{2n-2k}A_{2n-2k} (\vec{s} ), \qquad \text{with} \quad A_{2n}=I_{r},\\
G(\lambda,\vec{s})=\sum_{k=1}^{n}\frac{{\rm i}}{2}\lambda ^{2n-2k}G_{2n-2k} ( \vec{s} ), \\
E(\lambda,\vec{s})=\sum_{k=1}^{n} \frac{{\rm i}}{2}\lambda ^{2n-2k+1}E_{2n-2k+1} ( \vec{s} ) , \\
F(\lambda,\vec{s})=\sum_{k=1}^{n}\frac{{\rm i}}{2}\lambda ^{2n-2k+1}F_{2n-2k+1} ( \vec{s} ) .
\end{gather*}
\end{Proposition}
\begin{proof}
We start computing the logarithmic derivative of $ \Psi^{(n)} $ w.r.t.~$ S $, namely the quantity that we defined in~\eqref{eq:logderpsi} as
\begin{equation*}
\dS\Psi^{(n)}\big(\Psi^{(n)} \big) ^{-1}:=L^{(n)}.
\end{equation*}
The matrix-valued function $ L^{(n)} $ is entire in $ \lambda $, since it has no jumps along $ \gamma^{n}_{+} \cup \gamma^{n}_{-}. $ Furthermore, its asymptotic behavior at infinity is given by a matrix polynomial of degree $1$ in $ \lambda $. Thus, by the generalized Liouville theorem, we conclude that $ L^{(n)} $ is exactly a matrix polynomial of degree~$1$ in~$\lambda$.

In particular from the asymptotic expansion at $ \infty ,$ we find an explicit form of its matrix coefficients. Here and in the following series expansions in powers of~$\lambda$ we will use the notation $ [ \; \; ] _{\geq 0}$ to indicate that we are taking only the powers $\lambda^r$ with $r\geq 0$.
\begin{align*}
L^{(n)}(\lambda)&= \dS\Psi^{(n)}\big(\Psi^{(n)} \big) ^{-1}=\left[ \left(I_{2r } + \sum_{j \geq 1 } \frac{\Xi^{(n)}_{j}}{\lambda^{j}} \right) {\rm i}\lambda \hat{\sigma}_{3}\left( I_{2r } + \sum_{j \geq 1 } \frac{\Theta^{(n)}_{j}}{\lambda^{j}} \right)\right] _{\geq 0} \\
&={\rm i}\lambda \hat{\sigma}_{3}+ {\rm i} \big( \Xi^{(n)}_{1}\hat{\sigma}_{3}+\hat{\sigma}_{3}\Theta^{(n)}_{1} \big)
={\rm i}\lambda \hat{\sigma}_{3}+{\rm i} \big[ \Xi^{(n)}_{1},\hat{\sigma}_{3} \big]
={\rm i}\lambda \hat{\sigma}_{3}+2\beta^{(n)}_{1}\otimes\sigma_{1},
\end{align*}
where in the last two passages we used the fact that $ \Theta_{1}^{(n)}=-\Xi_{1} ^{(n)} $ and then the symmetry~\eqref{eq:symcoef}.

We can then consider the second quantity defined in \eqref{eq:logderpsi}, namely
\[ \frac{\partial}{\partial \lambda} \Psi^{(n)}\big(\Psi^{(n)} \big) ^{-1}=:M^{(n)}. \]
We use the same argument as for $ L^{(n)} $. Indeed, also $ M^{(n)} $ is entire in~$ \lambda $, since it has no jumps along $ \gamma^{n}_{+} \cup \gamma^{n}_{-} $. Its asymptotic behavior at infinity is given by a matrix polynomial of degree~$2n$ in~$\lambda$. We thus conclude, by the generalized Liouville theorem, that~$ M^{(n)} $ is exactly a matrix polynomial in $\lambda$ of degree~$2n$.
In particular from the asymptotic expansion at $ \infty $ we can find an explicit form of this matrix:
\begin{gather*}
M^{(n)}(\lambda) =\partial_{\lambda} \Psi^{(n)}\big(\Psi^{(n)} \big) ^{-1}\\
\hphantom{M^{(n)}(\lambda}{}
=\left[ \left(I_{2r } + \sum_{j \geq 1 } \frac{\Xi^{(n)}_{j}}{\lambda^{j}} \right)\left(\left( \frac{{\rm i}\lambda^{2n}I_{r}}{2}+{\rm i}S\right) \otimes \sigma_{3} \right) \left( I_{2r } + \sum_{j \geq 1 } \frac{\Theta^{(n)}_{j}}{\lambda^{j}} \right)\right] _{\geq 0}\\
\hphantom{M^{(n)}(\lambda}{}
 =\frac{{\rm i}\lambda^{2n}}{2}\hat{\sigma}_{3}+{\rm i}S \otimes \sigma_{3}+ \sum_{l=1}^{2n}\frac{{\rm i}\lambda^{2n-l}}{2} \! \underbrace{\left(\Xi_{l}^{(n)} \hat{\sigma}_{3}+\hat{\sigma}_{3}\Theta_{l}^{(n)} + \!\sum_{j\colon j+k=l}\!\Xi_{j}^{(n)}\hat{\sigma}_{3}\left(\sum_{k= 1}^{l-1}\Theta_{k}^{(n)} \right) \right) }_{=M^{(n)}_{2n-l}}.
\end{gather*}
In order to obtain the remaining part of the statement, we use the following lemma.
\begin{Lemma}
The coefficient of the term $ \lambda^{2n-l} $ in the matrix $ M^{(n)} $ is such that:
\begin{itemize}\itemsep=0pt
\item if $ l=2m $, then
\begin{equation*}
M^{(n)}_{2n-2m}=A_{2n-2m}(\vec{s})\hat{\sigma}_{3}+G_{2n-2m}(\vec{s})\hat{\sigma}_{2};
\end{equation*}
\item if instead $ l=2m-1 $, then
\begin{equation*}
M^{(n)}_{2n-2m+1}=E_{2n-2m+1}(\vec{s})\otimes I_{2r}+F_{2n-2m+1}(\vec{s})\hat{\sigma}_{1}.
\end{equation*}
\end{itemize}
\end{Lemma}
\begin{proof}
The proof is a direct consequence of the symmetry property that the asymptotics coefficients of $ \Xi^{(n)}$, $\Theta^{(n)} $ have. We start with the even case $ l=2m $. The coefficient of the term~$ \lambda^{2n-2m} $ in the matrix $ M^{(n)} $ is given by the following sum:
\begin{equation*}
M^{(n)}_{2n-2m}=\left(\Xi_{2m}^{(n)} \hat{\sigma}_{3}+\hat{\sigma}_{3}\Theta_{2m}^{(n)} + \sum_{j\colon j+k=2m}\Xi_{j}^{(n)}\hat{\sigma}_{3}\left(\sum_{k= 1}^{2m-1}\Theta_{k}^{(n)} \right) \right),
\end{equation*}
where in the last sum all the terms are of type
\[ \Xi_{2s}^{(n)}\hat{\sigma}_{3} \Theta_{2(m-s)}^{(n)} \qquad \text{or}\qquad \Xi_{2s-1}^{(n)}\hat{\sigma}_{3} \Theta_{2(m-s)+1}^{(n)}.\]
Using the symmetries~\eqref{eq:symcoef} and~\eqref{eq:symcoefinv}, a direct computation shows that these terms are always linear combinations of the Pauli's matrices $ \hat{\sigma}_{2}$, $\hat{\sigma}_{3} $.

So we can conclude that
\[M^{(n)}_{2n-2m} =A_{2n-2m}(\vec{s})\hat{\sigma}_{3}+G_{2n-2m}(\vec{s})\hat{\sigma}_{2} .\]
where the functions $ A_{2n-2m}(\vec{s})$, $G_{2n-2m}(\vec{s}) $ depend on the asymptotic coefficients of $ \Xi^{(n)}$, $\Theta^{(n)} $.

We work in the same way for the odd case, $ l=2m-1$. The coefficient of $ \lambda^{2n-2m+1} $ is given by the same formula
\begin{equation*}
M_{2n-2m+1}^{(n)}=\left(\Xi_{2m-1}^{(n)} \hat{\sigma}_{3}+\hat{\sigma}_{3}\Theta_{2m-1}^{(n)} + \sum_{j\colon j+k=2m-1}\Xi_{j}^{(n)}\hat{\sigma}_{3}\left(\sum_{k= 1}^{2m-2}\Theta_{k}^{(n)} \right) \right),
\end{equation*}
where in the last sum there are just terms of the two following types
\[ \Xi_{2s}^{(n)}\hat{\sigma}_{3} \Theta_{2(m-s)-1}^{(n)} \qquad \text{or}\qquad \Xi_{2s-1}^{(n)}\hat{\sigma}_{3} \Theta_{2(m-s)}^{(n)}.\]
In both of the cases, always replacing the symmetries \eqref{eq:symcoef} and \eqref{eq:symcoefinv}, they result to be linear combinations of $ I_{2r}$, $\hat{\sigma}_{1}$. Thus we can finally conclude that
\[ M_{2n-2m+1}^{(n)}=E_{2n-2m+1}(\vec{s})\otimes I_{2r}+F_{2n-2m+1}(\vec{s})\hat{\sigma}_{1}.\tag*{\qed} \]\renewcommand{\qed}{}
\end{proof}

Thanks to this lemma, the form of the matrix $ M^{(n)} $ is exactly the one of the statement and the proposition is completely proved.
\end{proof}

\begin{Remark}The system \eqref{eq:laxpairRH} for $ \Psi^{(n)} $ describes the isomonodromic deformations w.r.t.\ the deformation parameters $ s_{i}$, $i=1,\dots,r $, of the linear differential equation \[\frac{\partial}{\partial \lambda}\Psi ^{(n)}(\lambda, \vec{s})=M^{(n)}(\lambda, \vec{s})\Psi ^{(n)} ( \lambda, \vec{s} ), \]
that has only one irregular singular point at $ \infty $ of Poincar\'e rank $r= 2n+1, $ and in the special case of symmetry
\[ -\hat{\sigma}_{1}M^{(n)}(\lambda, \vec{s})\hat{\sigma}_{1}=M^{(n)}(-\lambda, \vec{s}). \]
\end{Remark}

We can finally state that the system \eqref{eq:laxpairRH} is an isomonodromic Lax pair for the matrix PII hierarchy~\eqref{eq:piinchier}.
\begin{Proposition}
For each fixed $ n $, the compatibility condition of the system \eqref{eq:laxpairRH}, i.e., the equation
\begin{equation}\label{eq:comp}
\frac{\partial}{\partial \lambda} L^{(n)}(\lambda,\vec{s})- \dS M^{(n)}(\lambda,\vec{s})+ \big[L^{(n)}(\lambda,\vec{s}), M^{(n)}(\lambda,\vec{s}) \big]=0
\end{equation}
is equivalent to the following equation
\begin{equation*}
\left( \dS + \acomw \right)\mathcal{L}_{n}\ubr= (-1)^{n+1}4^{n} [ S,W ] _{+},
\end{equation*}
Furthermore, the coefficients of the matrix $ M^{(n)} $ are written in terms of the matrix Lenard operators in the following way
\begin{gather}
A_{2n-2k}(\vec{s})=-\frac{1}{2}\left( -\frac{1}{4}\right)^{k-1}\! \left( \mathcal{L}_{k}\ubr - \left( \dS -\comw \dinv \comw \right)\!\left( \dS +\acomw \right)\mathcal{L}_{k-1} \ubr \right) , \nonumber\\
G_{2n-2k}(\vec{s})=\frac{{\rm i}}{2}\left(-\frac{1}{4} \right) ^{k-1}\left(\left( \dS -\comw \dinv \comw \right)\left( \dS +\acomw \right)\mathcal{L}_{k-1} \ubr \right), \nonumber\\
E_{2n-2k+1}\left( \vec{s}\right)=-{\rm i}\left(-\frac{1}{4} \right)^{k-1}\dinv\left( \comw\left( \acomw+\dS\right) \mathcal{L}_{k-1} \ubr \right) ,\nonumber\\
F_{2n-2k+1}\left( \vec{s}\right)=-{\rm i}\left(-\frac{1}{4} \right)^{k-1}\left(\left( \acomw +\dS\right) \mathcal{L}_{k-1} \ubr\right),\qquad
\text{for} \quad k=1,\dots, n.\label{eq:coeffM}
\end{gather}
In other words the system \eqref{eq:laxpairRH} is a Lax pair for the matrix Painlev\'e II hierarchy \eqref{eq:piinchier}.
\end{Proposition}
\begin{proof}We first rewrite the compatibility condition \eqref{eq:comp} as the following system of differential equations for the coefficients $ A$, $F$, $G$, $E $:
\begin{gather*}
\dS E(\lambda,\vec{s}) = [ W,F(\lambda,\vec{s}) ],\\
\dS A(\lambda,\vec{s}) =-{\rm i} [W,G(\lambda,\vec{s}) ] _{+},\\
\dS F(\lambda,\vec{s})=-2\lambda G(\lambda,\vec{s}) + [ W,E(\lambda,\vec{s}) ],\\
\dS G(\lambda,\vec{s})=2\lambda F(\lambda,\vec{s}) +{\rm i}[W,A(\lambda,\vec{s}) ] _{+} - [S,W] _{+}.
\end{gather*}
These equations must be satisfied identically in $ \lambda $. Thus, by the polinomiality of the coefficients $ A$, $F$, $G$, $E $, this system is equivalent to the following one
\begin{gather}
\dS E_{2n-2k+1}(\vec{s})= [ W, F_{2n-2k+1}(\vec{s}) ],\nonumber\\
\dS A_{2n}=0,\nonumber\\
\dS A_{2n-2k}(\vec{s})=-{\rm i} [W, G_{2n-2k}(\vec{s}) ] _{+},\nonumber\\
G_{2n-2k}(\vec{s})=\frac{1}{2}\left(-\dS F_{2n-2k+1}(\vec{s})+ [W,E_{2n-2k+1} (\vec{s}) ] \right),\nonumber \\
F_{2n-1}(\vec{s})=-\frac{{\rm i}}{2}\left[ W, A_{2n} \right] _{+},\nonumber\\
F_{2n-2k-1}(\vec{s})=\frac{1}{2}\left( \dS G_{2n-2k}(\vec{s})-{\rm i} [W, A_{2n-2k}(\vec{s}) ] _{+}\right),\nonumber \\
\frac{{\rm i}}{2}\dS G_{0}(\vec{s})=-[S,W]_{+}-\frac{1}{2} [W,A_{0}(\vec{s})] _{+}\qquad \text{for} \quad k=1,\dots, n.\label{eq:difeqcoeff}
\end{gather}
In order to prove the statement, we are going to prove by induction over $ l=2n-j $ that each coefficient $A_{2n-2k}$, $E_{2n-2k+1}$, $G_{2n-2k}$, $F_{2n-2k+1}$ is given by the formulae~\eqref{eq:coeffM} and that this implies that the last equation in the system~\eqref{eq:difeqcoeff} is exactly the $ n$-th member of the PII hierarchy~\eqref{eq:piinchier}.

We first check that for $ l=2n-1,2n-2$ the formulae~\eqref{eq:coeffM} are solutions of the equations~\eqref{eq:difeqcoeff}, i.e., the coefficients $ F_{2n-1}$, $E_{2n-1}$, $G_{2n-2}$, $A_{2n-2}$, are given by these formulae.

 Since $ A_{2n}=I_{r}$, the equation
\[\dS A_{2n}=0\]
is satisfied.
Then, the equation for $ F_{2n-1} $ will be satisfied for
\[F_{2n-1}=-{\rm i}W,\]
that is exactly the result of the formula in \eqref{eq:coeffM} for $ k=1 $, since
\[-{\rm i}\left(-\frac{1}{4} \right)^{0}\left(\left( \acomw +\dS\right) \mathcal{L}_{0} \ubr\right)=-{\rm i}W.\]
As a consequence, the equation for the coefficient $ E_{2n-1} $ in the system \eqref{eq:difeqcoeff} becomes
\[\dS E_{2n-1}(\vec{s})=0,\]
thus $ E_{2n-1} $ is constant w.r.t.\ the variable $S$ and it is in particular $E_{2n-1}=0 , $ because of the asymptotics of $ \Psi^{(n)}$.
This is also what is given by the formula for $ k=1 $:
\[-{\rm i}\left(-\frac{1}{4} \right)^{0}\dinv\left( \comw\left( \acomw+\dS\right) \mathcal{L}_{0} \ubr \right)=0.\]
We can then compute the term $ G_{2n-2} $ for which the equation in~\eqref{eq:difeqcoeff} is now
\[G_{2n-2}=-\frac{1}{2}\dS ( -{\rm i}W ) =\frac{{\rm i}}{2} W_{S},\]
that coincides with the formula
\[\frac{{\rm i}}{2}\left(-\frac{1}{4} \right) ^{0}\left(\left( \dS -\comw \dinv \comw \right)\left( \dS +\acomw \right)\mathcal{L}_{0} \ubr \right)=\frac{{\rm i}}{2}\dS W.\]
Finally, we can compute the term $ A_{2n-2} .$ It is supposed to satisfy, from the system \eqref{eq:difeqcoeff}, the equation
\[\dS A_{2n-2}=-{\rm i} [W, G_{2n-2} ] _{+}=\frac{1}{2} [W, W_{S} ] _{+}.\]
Integrating and taking the constant of integration another time equal $ 0 $ (for the same reason used above) we get
\[ A_{2n-2}= \frac{1}{2}W^{2}.\]
The same that is given by the formula{\samepage
\begin{gather*} -\frac{1}{2}\left( -\frac{1}{4}\right)^{0} \left( \mathcal{L}_{1}\ubr - \left( \dS -\comw \dinv \comw \right)\left( \dS +\acomw \right)\mathcal{L}_{0} \ubr \right)\\
\qquad{} =-\frac{1}{2}\big(W_{S}-W^{2}-W_{S} \big) . \end{gather*}
Thus for $ k=1 $ the formulas in \eqref{eq:coeffM} gives solutions of the system \eqref{eq:difeqcoeff}.}

Now we proceed by induction: supposing that for $l=2n-2k+1 $ the coefficients $ E_{2n-2k+1}$, $F_{2n-2k+1} $ are given by the formulas \eqref{eq:coeffM}, we will find that then also the coefficients for $l=2n-2k $ and $ l=2n-2k-1$ have the form given by the formulas \eqref{eq:coeffM}.

Indeed, from the equations in \eqref{eq:difeqcoeff}, we have
\begin{gather*}
G_{2n-2k}(\vec{s}) =\frac{1}{2}\left(-\dS F_{2n-2k+1}(\vec{s})+ [W,E_{2n-2k+1} (\vec{s}) ] \right) \\
\hphantom{G_{2n-2k}(\vec{s})}{}
=-\frac{1}{2}\left( -{\rm i}\left(-\frac{1}{4} \right)^{k-1}\dS\left(\left( \acomw +\dS\right) \mathcal{L}_{k-1} \ubr\right)\right) \\
\hphantom{G_{2n-2k}(\vec{s})=}{}+ \frac{1}{2}\left( [W,\cdot ] \left( -{\rm i}\left(-\frac{1}{4} \right)^{k-1}\dinv\left( \comw\left( \acomw+\dS\right) \mathcal{L}_{k-1} \ubr \right)\right) \right) \\
\hphantom{G_{2n-2k}(\vec{s})}{} =\frac{{\rm i}}{2}\left(-\frac{1}{4} \right)^{k-1}\left(\left( \dS -\comw \dinv \comw \right)\left( \dS +\acomw \right)\mathcal{L}_{k-1} \ubr \right)
\end{gather*}
that is exactly the formula in \eqref{eq:coeffM} for this coefficient.
Then we can compute
\begin{gather*}
A_{2n-2k}(\vec{s}) =-{\rm i}\dinv [W, G_{2n-2k}(\vec{s}) ] _{+} = \frac{1}{2}\left(-\frac{1}{4} \right)^{k-1}\dinv \acomw\\
 \hphantom{A_{2n-2k}(\vec{s})=}{}
\times \left(\left( \dS -\comw \dinv \comw \right)\left( \dS +\acomw \right)\mathcal{L}_{k-1} \ubr \right)\\
\hphantom{A_{2n-2k}(\vec{s})}{}
 = -\frac{1}{2}\left(-\frac{1}{4} \right)^{k-1}\!\left( \mathcal{L}_{k}\ubr - \left( \dS -\comw \dinv \comw \right)\!\left( \dS +\acomw \right)\mathcal{L}_{k-1} \ubr \right),
\end{gather*}
where in the last passage we have integrated (taking the integration's constant 0) after having applied formula \eqref{eq:factlennc}.
Then the equation for $ F_{2n-2k-1}(\vec{s})$ reads as
\begin{gather*}
F_{2n-2k-1} =\frac{1}{2}\left( \dS G_{2n-2k}(\vec{s})-{\rm i}\left[W, A_{2n-2k}(\vec{s}) \right] _{+}\right)\\
\hphantom{F_{2n-2k-1}}{}
 =\frac{1}{2}\left( \dS\frac{{\rm i}}{2}\left(-\frac{1}{4} \right)^{k-1}\left(\left( \dS -\comw \dinv \comw \right)\left( \dS +\acomw \right)\mathcal{L}_{k-1} \ubr \right) \right) \\
\hphantom{F_{2n-2k-1}=}{}
 -\frac{{\rm i}}{2}\Bigg(\acomw \frac{1}{2}\left(-\frac{1}{4} \right)^{k-1}\Bigg( \mathcal{L}_{k}\ubr - \left( \dS -\comw \dinv \comw \right)\\
\hphantom{F_{2n-2k-1}=}{}\times
 \left( \dS +\acomw \right)\mathcal{L}_{k-1} \ubr \Bigg)\Bigg) =-{\rm i}\left(-\frac{1}{4} \right)^{k}\left( \dS +\acomw \right)\mathcal{L}_{k}\ubr,
\end{gather*}
where in the last line we used another time property \eqref{eq:factlennc} of the matrix Lenard operators. Finally, the formula for $ E_{2n-2k-1} $ directly follows from the equation above and taking the integration contant equal~0, while integrating the equation~\eqref{eq:difeqcoeff}.

In the end, when we replace the formulas for $ G_{0}$, $A_{0} $ in the last equation of the system \eqref{eq:difeqcoeff}, namely
\[ \frac{{\rm i}}{2}\dS G_{0}(\vec{s})=-[S,W]_{+}-\frac{1}{2} [W,A_{0}(\vec{s}) ] _{+}, \]
using another time the property \eqref{eq:factlennc} we get the $n$-th member of the Painlev\'e II hierarchy:
\[ \left(\acomw + \dS \right) \mathcal{L}_{n} \ubr= ( -1 ) ^{n+1}4^{n} [S,W]_{+} .\tag*{\qed} \]\renewcommand{\qed}{}
\end{proof}

\begin{Remark}
The matrices $ L^{(n)}$, $M^{(n)} $ obtained here, are the analogue of the Lax pair for the scalar Painlev\'e II hierarchy obtained in \cite{clarkson2006lax}, with $W(\vec{s})$ given by
\begin{equation*}
2\beta_{1}^{(n)}(\vec{s})= -2{\rm i}\lim_{| \lambda | \rightarrow \infty}\big( \lambda\Xi^{(n)}(\vec{s})\big)_{1,2}:=W(\vec{s}).
\end{equation*}
\end{Remark}

We can then state and prove the final result of this study, that links solutions of the homogeneus matrix Painlev\'e~II hierarchy~\eqref{eq:piinchier} to Fredholm determinants of the matrix Airy convolution operators.
\begin{Corollary}\label{Corollary:solpiinc}
There exists a solution $W$ of the $ n$-th member of the PII hierarchy \eqref{eq:piinchier} connected to Fredholm determinant of the $n$-th Airy matrix operator \eqref{eq:freddet} through the following formula
\begin{equation}\label{eq: fred det and PII sol}
-\Tr\big( W^{2} (\vec{s})\big) = \dSS \ln \big( F^{(n)} (s_{1},\dots,s_{r} ) \big) .
\end{equation}
This solution $ W$ has boundary behavior $ ( W ) _{k,l=1}^r \sim -2 ( c_{kl}\mathrm{Ai}_{2n+1} (s_k + s_l) ) _{k,l=1}^r $ in the regime $\mathrm{s}\rightarrow + \infty$ with $ | \delta_j | \leq m $ for every $ j $, where $\mathrm{s} \coloneqq \frac{1}{r}\sum_{j=1}^{r}s_j$ is the baricenter of the variables $s_j$, and $\delta_j\coloneqq s_j - \mathrm{s}$.
\end{Corollary}
\begin{proof}
We first prove the formula \eqref{eq: fred det and PII sol}. The statement is achieved by Theorem~\ref{Proposition:ascoeffredolm} and the relation between $ \alpha_{1}^{(n)}$, $\beta_{1}^{(n)} $ given by
\begin{equation}\label{eq:relaab}
\dS \alpha_{1}^{(n)}=-2{\rm i}\big( \beta_{1}^{(n)}\big)^{2} .
\end{equation}
This relation holds for each $ n $ and it is obtained by looking at the coefficient of the term $ \lambda^{-1} $ in the asymptotic expansion at $ \infty $ of
\[\dS \Psi^{(n)} \big( \Psi^{(n)}\big) ^{-1},\]
and recalling that it must be 0. Indeed, from the asymptotic expansion of $ \Psi^{(n)} $ we have that the power $\lambda^{-1}$ coming from the formal asymptotic expansion of $\dS \Psi^{(n)} \big( \Psi^{(n)}\big) ^{-1}$ is\footnote{Here the notation $ [ \; \; ]_{-1} $ indicates that we only take the term $\lambda^{-1}$ in the relevant formal series.}
\begin{align*} \left[ \dS \Psi^{(n)}\big(\Psi^{(n)} \big) ^{-1}\right] _{-1} & =\left[ \left(I_{2r } + \sum_{j \geq 1 } \frac{\Xi^{(n)}_{j}}{\lambda^{j}} \right) {\rm i}\lambda \hat{\sigma}_{3}\left( I_{2r } + \sum_{j \geq 1 } \frac{\Theta^{(n)}_{j}}{\lambda^{j}} \right)\right] _{-1}\\
& = \frac{{\rm i}}{\lambda}\left( \Xi_{2}^{(n)}\hat{\sigma}_{3}+\hat{\sigma}_{3}\Theta_{2}^{(n)}+ \Xi_{1}^{(n)}\sigma_{3}\Theta_{1}^{(n)}+\dS \Xi_{1}\right) .
\end{align*}
And replacing in the coefficient of $ \lambda^{-1} $ the relations between the asymptotic coefficients of $\Theta^{(n)}$ and the ones of~$\Xi^{(n)}$, namely
\begin{gather*}
\Theta_{1}^{(n)}=-\Xi_{1}^{(n)} ,\qquad
\Theta_{2}^{(n)}=\big( \Xi_{1}^{(n)}\big) ^{2} -\Xi_{2}^{(n)}
\end{gather*}
the result is exactly the relation \eqref{eq:relaab}.

Now we are going to prove the second part of the statement. We define the scalar variables $\mathrm{s} \coloneqq \frac{1}{r}\sum_{j=1}^{r}s_j$ and $\delta_j\coloneqq s_j - \mathrm{s} $ for any $j=1,\dots, r$.

 We are now going to study the behavior of the solution $W$ for
\begin{equation}\label{eq:asimp. regime}
\mathrm{s}\rightarrow +\infty \qquad \text{and}\qquad | \delta_j | \leq m \quad \forall \, j.
\end{equation}
First, we rewrite the jump matrix $J^{(n)}(\lambda,\vec{s})$ of Riemann--Hilbert Problem~\ref{pb:rhp} in terms of the rescaled complex parameter $z\mathrm{s}^{\frac{1}{2n}}=\lambda $.

In particular we obtain that the jump matrices along $\gamma_+^n$ and along $\gamma_-^n$, are factorized in a product of commuting matrices, written in terms of the rescaled parameter~$z$ and the variables~$\mathrm{s}$,~$\delta_j$. Namely,
\begin{gather}
I_{2r}-2\pi {\rm i} r^{(n)}\big({\pm} z\mathrm{s}^{\frac{1}{2n}}\big)\otimes \sigma_{\pm}\chi_{\gamma_{\pm}^{n}}\big(z\mathrm{s}^{\frac{1}{2n}}\big) \nonumber\\
\qquad{} = \prod_{k,l=1}^{r}\left( I_{2r} + c_{kl} {\rm e}^{ \pm {\rm i}\mathrm{s}^{\frac{2n+1}{2n}}\left( \frac{z^{2n+1}}{2n+1}+z\left(2+\frac{\delta_k + \delta_l}{\mathrm{s}} \right) \right) } E_{k,l}\otimes \sigma_{\pm} \chi_{\tilde{\gamma}_{\pm}^n}(z)\right),\label{eq:factorization}
\end{gather}
where $E_{k,l}$ are the elementary matrices and $\sigma_{+}=\left(\begin{smallmatrix}
0&1\\0&0
\end{smallmatrix}\right)$, $\sigma_{-}=\left(\begin{smallmatrix}
0&0\\1&0
\end{smallmatrix}\right)$ and $\tilde{\gamma}^n_{\pm}$ are the transformed contours under the scaling $\lambda=z\mathrm{s}^{\frac{1}{2n}}$.

 Now, we are going to show that each matrix in the factorization \eqref{eq:factorization}, that we denote by $F_{kl}^{\pm}$, is close to the identity matrix $I_{2r}$ in the regime fixed in \eqref{eq:asimp. regime}. Remark that every $ F_{kl}^{\pm}$ has $2n$ critical points, corresponding to
\begin{equation*}
z_0^h = d_{kl}^{\frac{1}{2n}}{\rm e}^{{\rm i}\frac{\pi}{2n}(2h+1)} , \qquad h=0,\dots,2n-1,
\end{equation*}
where $d_{kl} = 2+\frac{\delta_k + \delta_l}{\mathrm{s}}$ is real, positive and bounded, while looking at the regime \eqref{eq:asimp. regime}.

We can then split the curves $\tilde{\gamma}_{\pm}^n$ respectively in the curves $\tilde{ \gamma}_{\pm,kl}^{n}$ one for each factor $F_{kl}^{\pm}$ appearing in the factorization \eqref{eq:factorization}. The curves $\tilde{ \gamma}_{\pm,kl}^{n}$ pass respectively through the points $ z_0^h$ with $h=0, \dots, n-1$ (in the upper plane) and $h=n,\dots,2n-1$ (in the lower plane).

 In this way, we can then evaluate the $\infty$-norm of each term $F_{kl}^{\pm}-I_{2r}$ and we have
\begin{equation*}
\big|\big| F_{kl}^\pm - I_{2r}\big|\big| _{\infty} = |c_{kl} | \sup_{z \in \tilde{ \gamma}_{\pm, kl}^n} {\rm e}^{\mp \mathrm{s}^\frac{2n+1}{2n} \Im \left( \frac{z^{2n+1}}{2n+1}+zd_{kl} \right)} = |c_{kl} |{\rm e}^{\mp\frac{2n}{2n+1}(\mathrm{s}d_{kl})^{\frac{2n+1}{2n} } \sin \left(\pm\frac{\pi}{2n}\right)}\rightarrow 0
\end{equation*}
for $\mathrm{s}\rightarrow+\infty$ and $|\delta_j| \leq m$ $\forall \, j$.

We can conclude that the rescaled jump matrix itself $J^{(n)}\big(\mathrm{s}^{\frac{1}{2n}} z\big)$ is close to the identity matrix in the regime~\eqref{eq:asimp. regime}, since each factor $F_{kl}^{\pm}$ in its factorization shares this property.

Consider now the rescaled function $X^{(n)}(z) \coloneqq \Xi^{(n)} \big(z\mathrm{s}^{\frac{1}{2n}}\big)$. By using Riemann--Hilbert Problem~\ref{pb:rhp} solved by $\Xi^{(n)}$, we have that
\begin{itemize}\itemsep=0pt
\item $X^{(n)}$ is analytic on $\mathbb{C}\setminus \tilde{ \gamma}_{+}^n\cup\tilde{ \gamma}_{-}^n$ and it admits continuous extension to these curves from either side;
\item its boundary values $X_{\pm}(z)$ while approaching $\tilde{ \gamma}_{+}^n\cup\tilde{ \gamma}_{-}^n$ from the left and respectively from the right, are related through the jump condition~\eqref{eq:jump} but with the rescaled jump matrix computed in~\eqref{eq:factorization};
\item for $|z| \rightarrow +\infty $ we have $X^{(n)} \sim I_{2r} + \sum_{j \geq 1 }\frac{X_j^{(n)}}{z^j}$.
\end{itemize}
Remark that we have $X_1^{(n)}=\mathrm{s}^{-\frac{1}{2n}}\Xi_1^{(n)}$.

By applying the small norm theorem (see for instance Theorem 1.5.1 in \cite{its2011large}), we conclude that the function $X^{(n)}(z)$ behaves as
\begin{equation}\label{eq:smallnormTheorem}
X^{(n)} (z ) = I_{2r} + \mathcal{O} \big( z^{-1} {\rm e}^{-C \mathrm{s}^{\frac{2n+1}{2n}}}\big),\qquad \mathrm{s}\rightarrow+\infty, \quad |\delta_j | \leq m \quad \forall \, j,
\end{equation}
for a certain value $C>0$.

Now, using the integral formula~\cite{its1990differential} for the rescaled solution of the Riemann--Hilbert Problem~\ref{pb:rhp}, namely $X^{(n)}(z)$, we have that
\begin{equation*}
X^{(n)}(z) = I_{2r}-\int_{\tilde{ \gamma}^{n}_{+}}\frac{ X^{(n)}_-(w ) r^{(n)}\big(w\mathrm{s}^{\frac{1}{2n}}\big)\otimes \sigma_+}{w-z} {\rm d}w -\int_{\tilde{ \gamma}^{n}_{-}}
\frac{X^{(n)}_-(w) r^{(n)}\big({-}w\mathrm{s}^{\frac{1}{2n}}\big)\otimes \sigma_-}{w-z} {\rm d}w,
\end{equation*}
and thus we recover the following expression for the first asymptotic coefficient
\begin{equation*}
\big( X_1^{(n)}\big)_{1,2} = \int_{\tilde{ \gamma}^{n}_{+}} X^{(n)}_-(w ) r^{(n)}\big(w\mathrm{s}^{\frac{1}{2n}}\big) {\rm d}w .
\end{equation*}
Finally, by recalling the definition of $W$ and using \eqref{eq:smallnormTheorem} we conclude that
\begin{equation*}
W = -2{\rm i}\big( \Xi_1^{(n)}\big)_{1,2}=- 2{\rm i}\mathrm{s}^{\frac{1}{2n}}\int_{\tilde{ \gamma}^{n}_{+}} X^{(n)}_-(w ) r^{(n)}\big(w\mathrm{s}^{\frac{1}{2n}}\big) {\rm d}w \sim -2 (c_{kl} \mathrm{Ai}_{2n+1}(s_k+s_l) ) _{k,l=1} ^r,
\end{equation*}
in the regime \eqref{eq:asimp. regime}.
\end{proof}
\begin{Remark}
Relation \eqref{eq: fred det and PII sol} can be thought as the noncommutative analogue of the results proved in \cite{tracy1994level} for the Painlev\'e II equation and in \cite{cafasso2019fredholm,le2018multicritical} for the scalar Painlev\'e II hierarchy, connecting the theory of Painlev\'e trascendents to the determinantal point processes theory. For the noncommutative Painlev\'e~II equation~\eqref{eq:ncpii}, this link was already established in~\cite{bertola2012fredholm} and here we actually extended that result to the noncommutative hierarchy~\eqref{eq:piinchier}.
\end{Remark}

\subsubsection*{Acknowledgements}
This project was supported by the European Union Horizon 2020 research and innovation program under the Marie Sklodowska-Curie RISE 2017 grant agreement no.~778010 IPaDEGAN. I~am also grateful to my supervisors Mattia Cafasso and Marco Bertola for their useful advises and guidance.

\pdfbookmark[1]{References}{ref}
\LastPageEnding

\end{document}